\newif\iffinal
\finalfalse	
\finaltrue 

\newif\ifmarek
\marektrue

\documentclass[reqno,twoside,11pt]{amsart}
\iffinal\else\usepackage[notref,notcite]{showkeys}\fi
\usepackage{cite}
\usepackage{lscape}
\usepackage[toc,title,page]{appendix}
\usepackage{subcaption}
\usepackage{amsmath}
\usepackage{amsmath,lipsum}
\usepackage{amsfonts}
\usepackage{amssymb}
\usepackage{bbm}
\usepackage{graphicx}
\usepackage{booktabs}
\usepackage{enumitem}
\usepackage{booktabs}
\usepackage{multirow}
\usepackage{tikz}
\usetikzlibrary{arrows, automata, positioning}
\graphicspath{{/Users/aubain2015/Documents/york/}}
\usepackage{verbatim}
\usepackage{color}
\usepackage{float}
\usepackage{dirtytalk}
\IfFileExists{epsf.def}{\input epsf.def}{\usepackage{epsf}}
\ifmarek
\IfFileExists{myowntimes.sty}{\usepackage{myowntimes}\usepackage{mathrsfs}}
	{\usepackage{times}\newcommand{\mathscr}{\mathcal}}

\usepackage{mathptmx}
\fi

\DeclareFontFamily{OT1}{eusb}{} \DeclareFontShape{OT1}{eusb}{m}{n}
{<5> <6> <7> <8> <9> <10> <11> <12> <14.4> eusb10}{}
\DeclareMathAlphabet{\eusb}{OT1}{eusb}{m}{n}

\DeclareFontFamily{OT1}{eusm}{} \DeclareFontShape{OT1}{eusm}{m}{n}
{<5> <6> <7> <8> <9> <10> <11> <12> <14.4> eusm10}{}
\DeclareMathAlphabet{\eusm}{OT1}{eusm}{m}{n}

\DeclareFontFamily{OT1}{eufm}{} \DeclareFontShape{OT1}{eufm}{m}{n}
{<5> <6> <7> <8> <9> <10> <11> <12> <14.4> eufm10}{}
\DeclareMathAlphabet{\mathfrak}{OT1}{eufm}{m}{n}

\DeclareFontFamily{OT1}{fraktura}{}
\DeclareFontShape{OT1}{fraktura}{m}{n} {<5> <6> <7> <8> <9> <10>
<11> <12> <13> <14.4> [1.1] eufm10}{}
\DeclareMathAlphabet{\fraktura}{OT1}{fraktura}{m}{n}

\DeclareFontFamily{OT1}{cmfi}{} \DeclareFontShape{OT1}{cmfi}{m}{n}
{<5> <6> <7> <8> <9> <10> <11> <12> <13> <14.4> [0.9] cmfi10}{}
\DeclareMathAlphabet{\cmfi}{OT1}{cmfi}{b}{n}

\DeclareFontFamily{OT1}{cmss}{} \DeclareFontShape{OT1}{cmss}{m}{n}
{<5> <6> <7> <8> <9> <10> <11> <12> <13> <14.4> cmss10}{}
\DeclareMathAlphabet{\cmss}{OT1}{cmss}{m}{n}

\newtheoremstyle{thm}{1.5ex}{1.5ex}{\itshape\rmfamily}{}
{\bfseries\rmfamily}{}{2ex}{}

\newtheoremstyle{def}{1.5ex}{1.5ex}{\slshape\rmfamily}{}
{\bfseries\rmfamily}{}{2ex}{}

\newtheoremstyle{rem}{1.3ex}{1.3ex}{\rmfamily}{}
{\itshape}
{} {1.5ex}{}

\theoremstyle{thm}
\newtheorem{theorem}{Theorem}[section]

\newtheorem{corollary}[theorem]{Corollary}

\theoremstyle{def}

\theoremstyle{rem}

\numberwithin{equation}{section}

\renewcommand{\subsection}{\secdef\subsct\sbsect}
\newcommand{\subsct}[2][default]{\refstepcounter{subsection}
\addcontentsline{toc}{subsection}
{{\tocsection{\!\!}{\hspace{1.2em}\thesubsection}{\!\!\!\!#1\dotfill}}{}}
\nopagebreak\vspace{0.45\baselineskip} {\flushleft\bf
\thesection.\arabic{subsection}~\bf #1.~}
\\*[3mm]\noindent
\nopagebreak}
\newcommand{\sbsect}[1]{\vspace{0.1cm}\noindent
\textbf{#1.~}\vspace{0.1cm}}

\renewcommand{\subsubsection}{%
\secdef \subsubsect\sbsbsect}
\newcommand{\subsubsect}[2][default]{%
\refstepcounter{subsubsection} 
\addcontentsline{toc}{subsubsection}{{\tocsection{\!\!}
{\hspace{3.05em}\thesubsubsection}{\!\!\!\!#1\dotfill}}{}}
\nopagebreak
\vspace{0.15\baselineskip} \nopagebreak {\flushleft\rmfamily
\itshape\arabic{section}.\arabic{subsection}.\arabic{subsubsection}
\ \rmfamily #1\/.}\ }
\newcommand{\sbsbsect}[1]{\vspace{0.1cm}\noindent
\rmfamily \itshape
\arabic{section}.\arabic{subsection}.\arabic{subsubsection} \
\sffamily #1\/.\ }
\iffinal

\else

\fi
   
\usepackage[toc,title,page]{appendix}

\newcommand{\scrF}{\mathscr{F}}

\title[]
{Bitcoin versus S\&P500 Index: Return and Risk Analysis}
\author [] {A.H.Nzokem$^1$}
\thanks {$^1$ hilaire77@gmail.com}

\begin{document}

\begin{abstract}
S$\&$P 500 index is considered the most popular trading instrument in financial markets. With the rise of cryptocurrencies over the past years, Bitcoin has also grown in popularity and adoption. The paper aims to analyze the daily return distribution of the Bitcoin and S$\&$P 500 index and assess their tail probabilities through two financial risk measures. As a methodology, We use Bitcoin and S$\&$P 500 Index daily return data to fit The seven-parameter General Tempered Stable (GTS) distribution using the advanced Fast Fractional Fourier transform (FRFT) scheme developed by combining the Fast Fractional Fourier (FRFT) algorithm and the 12-point rule Composite Newton-Cotes Quadrature. The findings show that peakedness is the main characteristic of the S$\&$P 500 return distribution, whereas heavy-tailedness is the main characteristic of the Bitcoin return distribution. The GTS distribution shows that $80.05\%$ of S$\&$P 500 returns are within $-1.06\%$ and $1.23\%$ against only $40.32\%$ of Bitcoin returns.  At a risk level ($\alpha$), the severity of the loss (\textbf{$AVaR_{\alpha}(X)$}) on the left side of the distribution is larger than the severity of the profit (\textbf{$AVaR_{1-\alpha}(X)$}) on the right side of the distribution.  Compared to the S$\&$P 500 index, Bitcoin has $39.73\%$ more prevalence to produce high daily returns (more than $1.23\%$ or less than $-1.06\%$). The severity analysis shows that at a risk level ($\alpha$) the average value-at-risk (\textbf{$AVaR(X)$}) of the bitcoin returns at one significant figure is four times larger than that of the S$\&$P 500 index returns at the same risk.
\end{abstract}

\maketitle
\section {Introduction}
\noindent
A cryptocurrency is a cryptographically secured digital asset, sometimes known as a cryptoasset \cite{lewis2018basics}. Bitcoin was the first cryptocurrency created in 2009 by Satoshi Nakamoto. The idea behind Bitcoin was to create 
a peer-to-peer electronic payment system that allows online payments to be sent directly from one party to another without going through a financial institution\cite{nakamoto2008bitcoin}. Bitcoin was set up to replace financial institutions with a payment network based on the blockchain or distributed ledger technology. Since its inception, Bitcoin has grown in popularity and adoption and is now viewed as a viable legal tender in some countries. Its rising popularity has attracted growing interest and questions in economics and finance literature regarding the usage of Bitcoin as currency and the formation and dynamic of Bitcoin prices.\\
\noindent
The economic appraisal of Bitcoin as currency has been done in many studies\cite{baur2021volatility,baur2016bitcoin,bjerg2016bitcoin,yermack2015bitcoin,stephanie2014bitcoin}. Like other traditional fiat currencies (Dollar, Euro, Yen, \dots), whether the bitcoin may be considered as a currency depends on its ability to fulfill the three basic functions: A medium of exchange, a store of value, and a unit of account. As a medium of exchange, Bitcoin can be used to pay someone for something or to extinguish a debt or financial obligation. However, bitcoin bears exchange rate risk (BTC/USD)\cite{baur2021volatility, baur2016bitcoin} and is not widely accepted in its current state; only five out of the top 500 online merchants accepted Bitcoin in 2016 \cite{lewis2018basics}. As a store of value, Bitcoin will be worth the same as it is today. This function is also rejected in the literature \cite{baur2021volatility, baur2016bitcoin, yermack2015bitcoin, stephanie2014bitcoin} as Bitcoin is unstable and has excess volatility.
As a unit of count, Bitcoin can be used to compare the value of two items or to count up the total value of assets. The extreme volatility\cite{baur2018bitcoin} of Bitcoin makes it difficult or impossible to derive the true value of a specific good in Bitcoin. Some studies \cite{baur2018bitcoin,yermack2015bitcoin} show that the statistical properties of Bitcoin are uncorrelated with traditional asset classes such as stocks, bonds, and commodities, and the transaction analysis of Bitcoin accounts shows that Bitcoins are mainly used as an investment tool and not as a currency. A similar study \cite{baur2021volatility} shows that the volatility of Bitcoin prices is extreme and almost ten times higher than the volatility of major exchange rates (US/Euro and US/Yen) and concludes that Bitcoin cannot function as a medium of exchange, but can be used as a risk-diversified tool.\\
\noindent
The formation and dynamics of Bitcoin prices is another important aspect of the bitcoin that is studied in the literature \cite{ciaian2016economics,koutmos2023investor,chen2020fear,hakim2020bitcoin}. Several factors affecting Bitcoin price have been identified in the literature review: market forces of Bitcoin
supply and demand, Bitcoin attractiveness for investors, and global macroeconomic and financial development. The empirical results \cite{ciaian2016economics} show that the market forces of Bitcoin supply and demand greatly impact Bitcoin price, confirming the major role played by the standard economic model of currency in explaining Bitcoin price formation. However, the same study \cite{ciaian2016economics} shows that the global macro-financial development (captured by the Dow Jones Index, exchange rate, and oil price) does not significantly impact the Bitcoin price. The relationship between Bitcoin price and attractiveness has also been studied \cite{koutmos2023investor, hakim2020bitcoin, chen2020fear}. The attractiveness variables are the number of Google searches that used the terms bitcoin, bitcoin crash, and crisis. The findings\cite{hakim2020bitcoin} show that a Bitcoin price increase is usually preceded by an increase in the worldwide interest in Bitcoin; similarly, a fallen Bitcoin price goes along with an increase in market mistrust over a collapse of Bitcoin. In addition, another study \cite{chen2020fear} shows that the coronavirus fear sentiment, captured by hourly Google search queries on coronavirus-related words, negatively impacted the Bitcoin price (negative Bitcoin returns and high trading volume).\\
\noindent
The above discussion shows that Bitcoin is mainly used as an investment tool, not a currency. The main determinants of the Bitcoin price are market forces of Bitcoin supply and demand and Bitcoin attractiveness for investors and users. We will analyze the daily return distribution of the bitcoin and S$\&$P 500 index and assess and compare their tail probabilities through two financial risk measures: the value-at-risk (VaR) and the average value-at-risk (AVaR). The findings will provide another perspective in understanding the distribution of the return and volatility of the bitcoin. As a methodology, We use Bitcoin and S$\&$P 500 Index daily return data over the 2010-2023 period to fit the General Tempered Stable (GTS) distribution to the underlying data return distribution. The GTS distribution is a seven-parameter family of infinitely divisible distribution, which covers several well-known distribution subclasses like Variance Gamma distributions \cite{madan1998variance, nzokem2022,nzokem_2021b,Nzokem_Montshiwa_2023}, bilateral Gamma distributions \cite{kuchler2013tempered, nzokem2023european} and CGMY distributions \cite{carr2003stochastic}. The main disadvantage of the GTS distribution is the lack of the closed form of the density, cumulative functions, and their derivatives. we use the numerical computations, which are based on the advanced Fast Fractional Fourier transform (FRFT) developed by combining the classic Fast Fractional Fourier (FRFT) algorithm and the 12-point rule Composite Newton-Cotes Quadrature.\\ 
\noindent 
We organized the paper as follows. Section 2 provides the trend and the volatility of Bitcoin and SP 500 Index daily price data. Section 3 briefly presents the GTS distribution’s theoretical framework. Section 4 develops the advanced Fast Fractional Fourier transform (FRFT) scheme. Section 5 shows the results of the GTS Parameter Estimation from daily return data. Section 6 analyses the probability density functions and some key statistics. Section 7 develops the methodology and computes the value-at-risk and average value-at-risk.

\section{Bitcoin and S\&P 500: Overview \& Trends}
\noindent 
The Standard $\&$ Poor’s 500 Composite Stock Price Index, also known as the S$\&$P 500, is a stock index that tracks the share prices of 500 of the largest public companies with stocks listed on the New York Stock Exchange (NYSE) and the Nasdaq in the United States. It was introduced in 1957 and often treated as a proxy for describing the overall health of the stock market or even the U.S. economy. The S$\&$P 500 data were extracted from Yahoo Finance. The historical prices span from April 28, 2013, to June 22, 2023, and were adjusted for splits and dividends.\\
The Bitcoin BTC price was extracted from CoinMarketCap and spanned from January 04, 2010, to June 16, 2023. The Bitcoin and S$\&$P 500 prices are denominated in US dollars, representing the currency in which Bitcoin and S$\&$P 500 are the most traded.

\subsection{Bitcoin’s Path Compared with the S$\&$P 500 Index's path}
\noindent
Bitcoin's price fluctuations primarily stem from investors and traders. Investors use Bitcoin to store value, generate wealth, and hedge against inflation, whereas traders use it to bet against its price changes. The Bitcoin historical price used in Fig \ref{fig01} starts on April 28, 2013, when Bitcoin was trading at $\$134$. By December 2013, the price had spiked to $\$1,151$ for the first and then fallen to $\$698$ three days later. In 2017, Bitcoin's price hovered around $\$1,000$ until it broke $\$2,000$ after mid-May and then skyrocketed to $\$19,497.40$ on Dec 16, 2027 (first peak in Fig \ref{fig01}). The increasing demand for Bitcoin has triggered the development of cryptocurrencies to compete with Bitcoin. In 2020, the COVID-19 pandemic resulted in massive disruption in the global economy. The government policies combined with investor and trader responses accelerated the dynamic of the Bitcoin price. The bitcoin price started at $\$7,200$ in early 2020; on November 30, 2020, Bitcoin was trading for $\$19,625$, and the price reached around $\$29,000$ at the end of December 2020. 
\begin{figure}[ht]
\vspace{-0.3cm}
    \centering
\hspace{-1cm}
  \begin{subfigure}[b]{0.45\linewidth}
    \includegraphics[width=\linewidth]{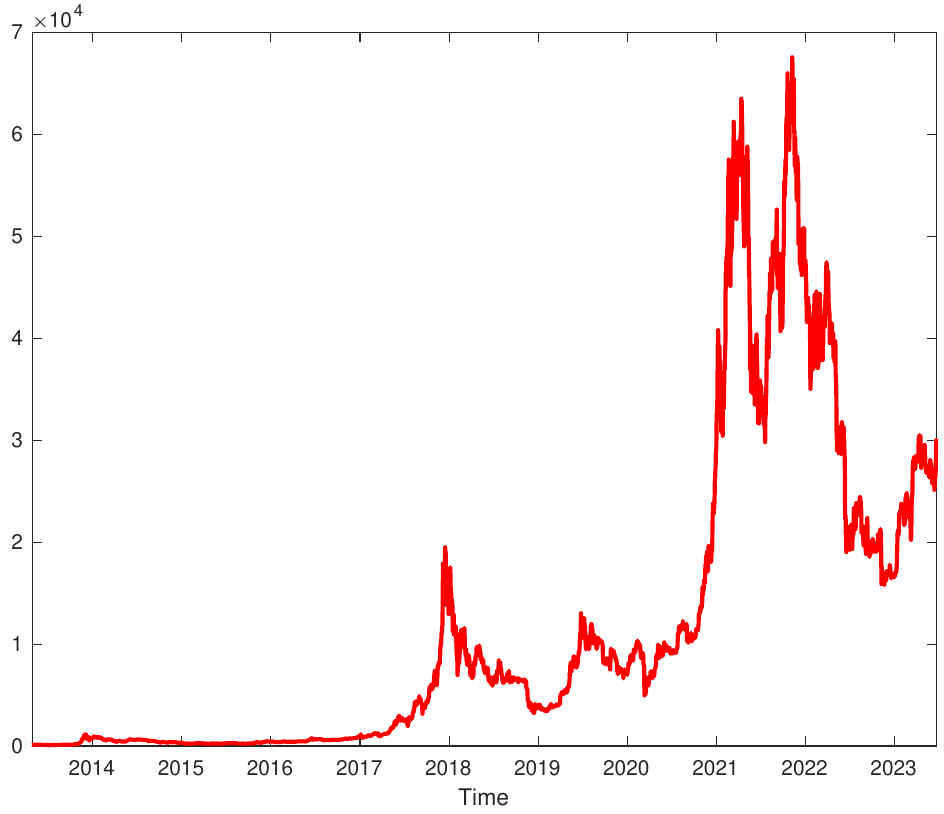}
\vspace{-0.6cm}
     \caption{Bitcoin Daily Price}
         \label{fig01}
  \end{subfigure}
  \begin{subfigure}[b]{0.45\linewidth}
    \includegraphics[width=\linewidth]{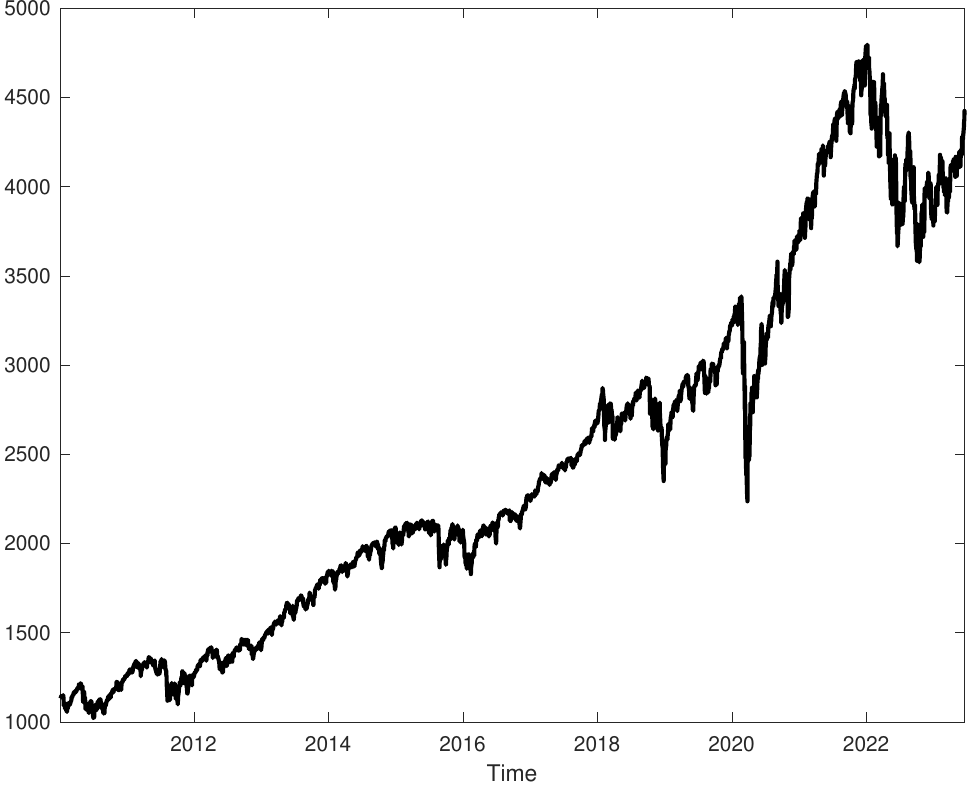}
\vspace{-0.6cm}
     \caption{S\&P 500 Daily Price}
         \label{fig02}
          \end{subfigure}
\vspace{-0.6cm}
  \caption{Daily Price}
  \label{fig1}
\vspace{-0.6cm}
\end{figure}
\noindent
The Bitcoin price increased by $14\%$ in January 2021 and reached $\$33,114$. By mid-April, Bitcoin prices reached the highest peak of $\$63,503$ on April 13, 2021 (second peak in Fig \ref{fig01}) before starting a decreasing process to reach $\$29,807$. On November 20, 2021, Bitcoin again achieved the highest value of $\$65,995$. In 2022, Bitcoin's price declined gradually, with prices only reaching $\$19,784$ at the end of June before falling further to $\$16547$ at the end of December. In 2023, Bitcoin price increased gradually and reached the value of $\$30,027$ on June 21, 2023 (third peak in Fig \ref{fig01}).\\
\noindent
In Contrast to Bitcoin price, Fig \ref{fig02} shows that S$\&$P 500 index smooth out price fluctuations. The S$\&$P 500 historical price used in Fig \ref{fig02} starts on January 04, 2010, when S$\&$P 500 index worth at $\$1132$. From 2010 to 2019, the U.S. economy was characterized by Stable economic growth and low-interest rates, which helped to keep equity prices on the rise. The S$\&$P 500 index has steadily increased and almost tripled the index value from $\$1132$ to $\$3230$ on December 31, 2019. In early 2020, many countries issued quarantines in which businesses were ordered to shut down due to the global spread of COVID-19; the negative impact of such policy on the economy has sent the S$\&$P 500 index into a tailspin, as shown in Fig \ref{fig02}. In fact, on February 19, 2020, the S$\&$P 500 closed at $\$3,386$, which was the highest value at that time. By March 23, 2020, the index plummeted to $\$2,237$, losing $34\%$ of its value. The S$\&$P 500 recovered the loss and continued a positive trend into 2021, reaching a peak of $\$4726$ on January 12, 2022. The index started a decreasing process and fell at the lower value of $\$3577$ on October 12, 2022, before increasing slightly to $\$4409$ on June 16, 2023.

\subsection{Bitcoin and S$\&$P 500 Index volatility over time}
\noindent
The volatility of bitcoin and the stock market index (S$\&$P 500) are produced and analyzed. The realized volatility measures the magnitude of daily price movements, regardless of direction, of some underlying, over a specific period. The realized volatility formula is provided in (\ref{eq:l31}) with T=month for the short term and T=year for the long term. Let the number of observations $m$, and the daily observed price $S_{j}$ on day $t_{j}$ with $j=1,\dots,m$.We have the following equations.
\begin{align}
 y_{j}=\log(S_{j}/S_{j-1}) \hspace{5mm}  \hbox{ $j=2,\dots,m$}  \quad \quad  \sigma^{2}_{k}=\frac{252}{T}\sum_{j=0}^{T} y^{2}_{k-j}  \hspace{5mm}  \hbox{ $k=T,\dots,m$} \label{eq:l31}
 \end{align}

\begin{figure}[ht]
\vspace{-0.3cm}
    \centering
\hspace{-1.5cm}
  \begin{subfigure}[b]{0.45\linewidth}
    \includegraphics[width=\linewidth]{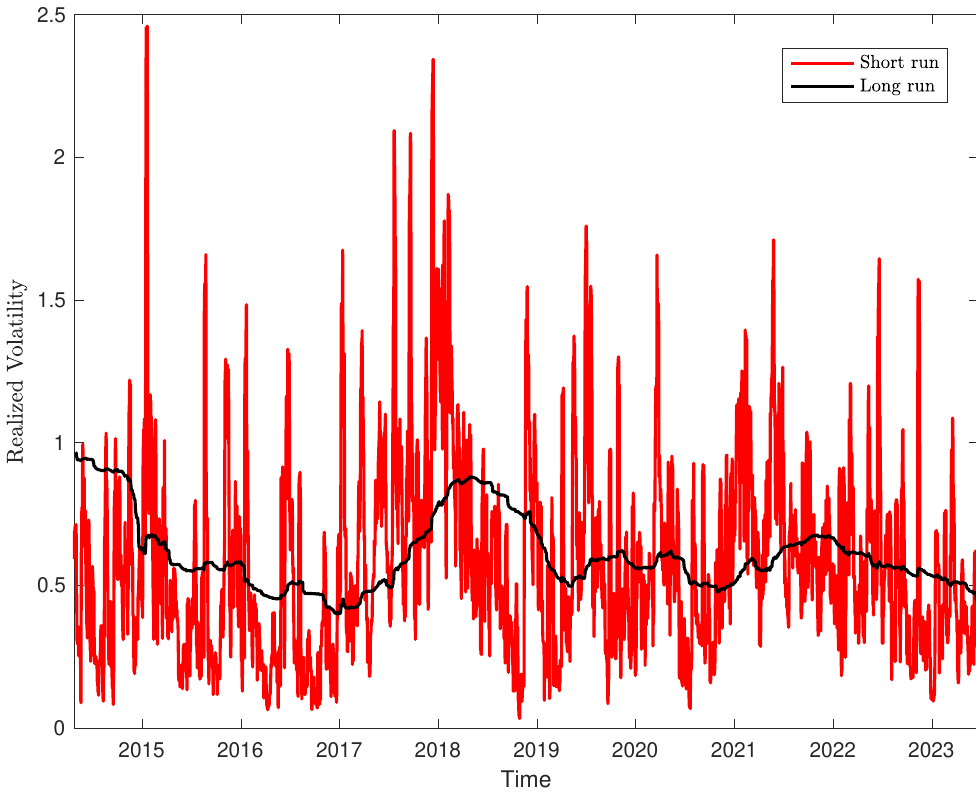}
\vspace{-0.6cm}
     \caption{Bitcoin Realized Volatility}
         \label{fig03}
  \end{subfigure}
\hspace{-0.4cm}
  \begin{subfigure}[b]{0.45 \linewidth}
    \includegraphics[width=\linewidth]{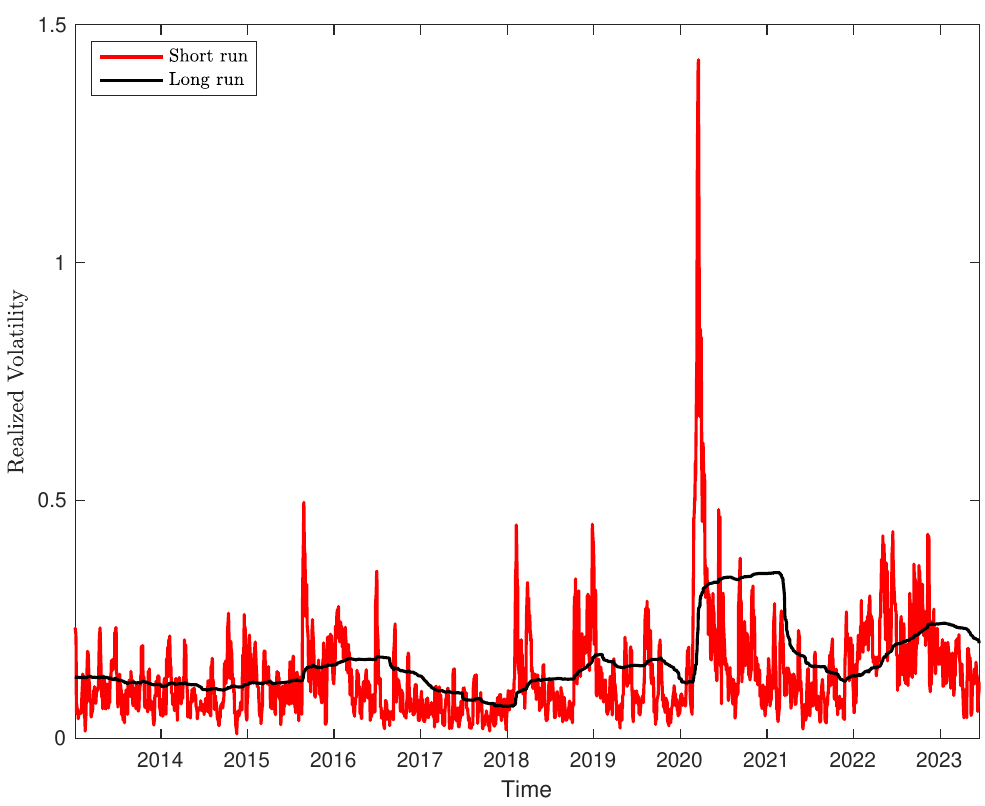}
\vspace{-0.6cm}
     \caption{S\&P 500 Index Realized Volatility}
         \label{fig04}
          \end{subfigure}
\vspace{-0.6cm}
  \caption{Realized Volatility}
  \label{fig3}
\vspace{-0.3cm}
\end{figure}

\noindent
The volatility represents both risk and opportunity for financial investments. Fig (\ref{fig03}) shows that Bitcoin is a highly volatile asset, and the volatility persists over time. The highly volatile security quickly hits new highs and lows and has rapid growths and sharp falls. In contrast, the S$\&$P 500 index is a lowly volatility asset. It has relatively stable price dynamics in the long run, as shown in Fig(\ref{fig04}). According to statistics used in Fig (\ref{fig3}), on average, Bitcoin is almost five times more volatile than the S$\&$P 500 index in the short run but four times more volatile in the long run.

\section{ Generalized Tempered Stable (GTS) Process: Overview}
\noindent
The L\'evy measure of the Generalized Tempered Stable (GTS) distribution ($V(dx)$) is defined (\ref{eq:l23}) as a product of a tempering function ($q(x)$) (\ref{eq:l21}) and a L\'evy measure of the $\alpha$-stable distribution ($V_{stable}(dx)$)(\ref{eq:l22}).
 \begin{align}
q(x) &= e^{-\lambda_{+}x} \boldsymbol{1}_{x>0} + e^{-\lambda_{-}|x|} \boldsymbol{1}_{x<0} \label{eq:l21}\\
V_{stable}(dx) &=\left(\frac{\alpha_{+}}{x^{1+\beta_{+}}} \boldsymbol{1}_{x>0} + \frac{\alpha_{-}}{|x|^{1+\beta_{-}}} \boldsymbol{1}_{x<0}\right) dx \label{eq:l22}\\
V(dx) =q(x)V_{stable}(dx)&=\left(\frac{\alpha_{+}e^{-\lambda_{+}x}}{x^{1+\beta_{+}}} \boldsymbol{1}_{x>0} + \frac{\alpha_{-}e^{-\lambda_{-}|x|}}{|x|^{1+\beta_{-}}} \boldsymbol{1}_{x<0}\right) dx \label{eq:l23}
 \end{align}
\noindent 
where $0\leq \beta_{+}\leq 1$, $0\leq \beta_{-}\leq 1$, $\alpha_{+}\geq 0$, $\alpha_{-}\geq 0$, $\lambda_{+}\geq 0$ and $\lambda_{-}\geq 0$. \\
The six parameters that appear have important interpretations.
 $\beta_{+}$ and $\beta_{-}$ are the indexes of stability bounded below by 0 and above by 2 \cite{borak2005stable}. They capture the peakedness of the distribution similarly to the $\beta$-stable distribution, but the distribution tails are tempered. If $\beta$ increases (decreases), then the peakedness decreases (increases). $\alpha_{+}$ and $\alpha_{-}$ are the scale parameters, also called the process intensity \cite{boyarchenko2002non}; they determine the arrival rate of jumps for a given size. $\lambda_{+}$ and $\lambda_{-}$ control the decay rate on the positive and negative tails. Additionally, $\lambda_{+}$ and $\lambda_{-}$ are also skewness parameters. If $\lambda_{+}>\lambda_{-}$ ($\lambda_{+}<\lambda_{-}$), then the distribution is skewed to the left (right), and if $\lambda_{+}=\lambda_{-}$, then it is symmetric \cite{rachev2011financial}. $\alpha$ and $\lambda$ are related to the degree of peakedness and thickness of the distribution. If $\alpha$ increases (decreases), the peakedness and the thickness decrease (increase). Similarly, If $\lambda$ increases (decreases), then the peakedness increases (decreases) and the thickness decreases (increases) \cite{bianchi2019handbook}.\\
\noindent
The GTS distribution can be denoted by $X\sim GTS(\textbf{$\beta_{+}$}, \textbf{$\beta_{-}$}, \textbf{$\alpha_{+}$},\textbf{$\alpha_{-}$}, \textbf{$\lambda_{+}$}, \textbf{$\lambda_{-}$})$ and $X=X_{+} - X_{-}$ with $X_{+} \geq 0$, $X_{-}\geq 0$. $X_{+}\sim TS(\textbf{$\beta_{+}$}, \textbf{$\alpha_{+}$},\textbf{$\lambda_{+}$})$ and $X_{-}\sim TS(\textbf{$\beta_{-}$}, \textbf{$\alpha_{-}$},\textbf{$\lambda_{-}$})$. 
\begin{align}
 \int_{-\infty}^{+\infty} V(dx) =\begin{cases}
  +\infty  &\text{if }{\beta_{+}\ge 0\vee\beta_{-} \ge 0}   \\
  \alpha_{+}{\lambda_{+}}^{\beta_{+}}\Gamma(-\beta_{+}) +  \alpha_{-}{\lambda_{-}}^{\beta_{-}}\Gamma(-\beta_{-}) &\text{if }{\beta_{+} < 0\wedge\beta_{-} < 0} \end{cases} \label{eq:l24}
     \end{align} 
From (\ref{eq:l24}), it results that when $\beta_{+} < 0$, TS(\textbf{$\beta_{+}$}, \textbf{$\alpha_{+}$},\textbf{$\lambda_{+}$}) is of finite activity and can be written as a Compound Poisson process on the right side ($X_{+}$). we have similar pattern when $\beta_{-} < 0$. However, when $0\le \beta_{+} \le 1$, $X_{+}$ is an infinite activity process with infinite jumps in any given time interval. We have a similar pattern when $0 \le \beta_{-} \le 1$. In addition to the infinite activities process, we have 
\begin{align}
  \int_{-\infty}^{+\infty} min(1,|x|)V(dx) <+ \infty \label{eq:l25} 
\end{align}
\noindent
By adding the location parameter, the GTS distribution becomes GTS(\textbf{$\mu$}, \textbf{$\beta_{+}$}, \textbf{$\beta_{-}$}, \textbf{$\alpha_{+}$},\textbf{$\alpha_{-}$}, \textbf{$\lambda_{+}$}, \textbf{$\lambda_{-}$}) and we have:
 \begin{align}
Y=\mu + X=\mu + X_{+} - X_{-} \quad \quad  Y\sim GTS(\mu, \beta_{+}, \beta_{-}, \alpha_{+}, \alpha_{-},\lambda_{+}, \lambda_{-}) \label {eq:l26}
  \end{align}

\begin{theorem}\label{lem5} \ \\ 
Consider a variable $Y \sim GTS(\textbf{$\mu$}, \textbf{$\beta_{+}$}, \textbf{$\beta_{-}$}, \textbf{$\alpha_{+}$},\textbf{$\alpha_{-}$}, \textbf{$\lambda_{+}$}, \textbf{$\lambda_{-}$})$, the characteristic exponent can be written
  \begin{align}
\Psi(\xi)=\mu\xi i + \alpha_{+}\Gamma(-\beta_{+})\left((\lambda_{+} - i\xi)^{\beta_{+}} - {\lambda_{+}}^{\beta_{+}}\right) + \alpha_{-}\Gamma(-\beta_{-})\left((\lambda_{-} + i\xi)^{\beta_{-}} - {\lambda_{-}}^{\beta_{-}}\right) \label {eq:l27}
  \end{align}
\end{theorem} 
See \cite{nzokem2022fitting} for theorem \ref{lem5} proof \\

 \begin{theorem}\label{lem7} (Cumulants $\kappa_{k}$)\\
 Consider a variable $Y \sim GTS(\textbf{$\mu$}, \textbf{$\beta_{+}$}, \textbf{$\beta_{-}$}, \textbf{$\alpha_{+}$},\textbf{$\alpha_{-}$}, \textbf{$\lambda_{+}$}, \textbf{$\lambda_{-}$})$. The cumulants $\kappa_{k}$ of the Generalized Tempered Stable distribution are defined as follows.
  \begin{align}
\kappa_{1}= \mu + \alpha_{+}{\frac{\Gamma(1-\beta_{+})}{\lambda_{+}^{1-\beta_{+}}}} - \alpha_{-}{\frac{\Gamma(1-\beta_{-})}{\lambda_{-}^{1-\beta_{-}}}} \quad
 \kappa_{k}=\alpha_{+}{\frac{\Gamma(k-\beta_{+})}{\lambda_{+}^{k-\beta_{+}}}} + (-1)^{k} \alpha_{-}{\frac{\Gamma(k-\beta_{-})}{\lambda_{-}^{k-\beta_{-}}}} \label {eq:l28}  \end{align}
 \end{theorem}
 See \cite{nzokem2022fitting} for theorem \ref{lem7} proof \\
 
 \noindent
The characteristic function, Fourier Transform ($F(f)$) and density function ($f$) of the GTS process $Y$ can be written as follows
\begin{align} 
F[f](\xi) =e^{\Psi(-\xi)} \quad \quad   f(y)= \frac{1}{2\pi}\int_{-\infty}^{+\infty}e^{iy x}F[f](x)dx  \label {eq:l3}
 \end{align}
Given the parameters $V=(\textbf{$\mu$}, \textbf{$\beta_{+}$}, \textbf{$\beta_{-}$}, \textbf{$\alpha_{+}$},\textbf{$\alpha_{-}$}, \textbf{$\lambda_{+}$}, \textbf{$\lambda_{-}$})$, we have the first and second order derivative of the function can be deduced from equation (\ref{eq:l3}).
\begin{align} 
 \frac{df(y)}{dV_{j}}= \frac{1}{2\pi}\int_{-\infty}^{+\infty}e^{iy x} \frac{dF(x)}{dV_{j}}dx \quad \quad  \frac{d^{2}f(y)}{dV_{k}dV_{j}}= &\frac{1}{2\pi}\int_{-\infty}^{+\infty}e^{iy x} \frac{d^{2}F(x)}{dV_{k}dV_{j}}dx \label {eq:l3}
 \end{align}
 
\section{Methodology: Review of Advanced FRFT - Based Scheme}
\subsection{Fast Fourier Transform and Fractional Fourier Transform }
\noindent
The Conventional fast Fourier transform (FFT) algorithm is widely used to compute discrete convolutions, discrete Fourier transforms (DFT) of sparse sequences, and to perform high-resolution trigonometric interpolation \cite {bailey1991fractional}. The discrete Fourier transforms (DFT) are based on n-th roots of unity $e^{-\frac{2\pi i}{n}}$. The generalization of DFT is the fractional Fourier transform, which is based on fractional roots of unity $e^{- 2\pi i\alpha}$, where $\alpha$ is an arbitrary complex number.\\
\noindent
The fractional Fourier transform is defined on $m$-long sequence ($x_1$, $x_ {2} $, \dots, $x_{m}$) as follows

 \begin{equation}
 G_{k+s}(x,\delta)=\sum_{j=0}^{m-1}\! x_{j}e^{-2\pi i j(k+s)\delta} \hspace{5mm}
   \hbox{$0\leq k<m$ \& $0\leq s\leq 1$}
   \label {eq:l31}
\end{equation}  
\noindent
Let us have $2j(k+s)=j^2 + (k+s)^2  - (k-j+s)^2$ , equation (\ref{eq:l31}) becomes
 \begin{equation}
 \begin{aligned}
 G_{k+s}(x,\delta)&=\sum_{j=0}^{m-1}\! x_{j}e^{-\pi i (j^2 + (k+s)^2  - (k-j+s)^2) \delta}  =e^{-\pi i (k+s)^2 }\sum_{j=0}^{m-1}\! x_{j}e^{-\pi i j^2 \delta}e^{\pi i (k-j+s)^2 \delta}\\\label {eq:l4}
&=e^{-\pi i (k+s)^2 }\sum_{j=0}^{m-1}\! y_{j}z_{k-j} \hspace{10mm}
   \hbox{$y_{j}=x_{j}e^{-\pi i j^2 \delta}$ \& $z_{j}=e^{\pi i (j+s)^2 \delta}$}
  \end{aligned}
\end{equation}
\noindent
The expression $ \sum_{j=0}^{m-1}\! y_{j}z_{k-j}$ is a discrete convolution. Still, we need a circular convolution (i.e., $z_{k-j}=z_{k-j+m}$) to evaluate $G_{k+s}(x,\delta)$. The conversion from discrete convolution to discrete circular convolution is possible by extending the sequences $y$ and $z$ to length $2m$ defined as follows.
\begin{equation}
 \begin{aligned}
y_j &= e^{-2\pi ij^2\delta}   &\quad z_j &= e^{-2\pi i(j+s)^2\delta}      &\quad  \hspace{5mm}  \hbox{$0\leq k<m$}  \label{eq:l5}\\
y_j&=0   &\quad z_j&=e^{-2\pi i(j+s-2m)^2\delta}  &\quad  \hspace{5mm}
   \hbox{$m\leq k<2m$}
  \end{aligned}
\end{equation}

\noindent
Taking into account the 2m-long sequence, the previous fractional Fourier transform becomes
 \begin{equation}
 G_{k+s}(x,\delta)=e^{-\pi i (k+s)^2 }\sum_{j=0}^{2m-1}\! y_{j}z_{k-j}=e^{-\pi i (k+s)^2 }{DFT}_k^{-1}[{{DFT}_j(y){DFT}_j(z)}]    \label {eq:l6}
\end{equation}  
\noindent
This procedure is referred to in the literature \cite {bailey1991fractional} as the Fast fractional Fourier Transform Algorithm with a total computational cost of $20mlog_{2}m + 44m$ operations.
\subsection{Fast Fourier Transform Algorithm and Direct Integration Methods} 
\noindent
We assume that $\scrF[f](y)$ is zero outside the interval $[-\frac{a}{2}, \frac{a}{2}]$; $\beta=\frac{a}{m} $ is the step size of the $m$ input values of $\scrF[f](y)$, defined by $y_{j}=(j-\frac{m}{2}) \beta$ for $ 0 \leq j <m$. Similarly, $\gamma$ is the step size of the $m$ output values of $f(t)$, defined by $x_{k}=(k-\frac{m}{2}) \gamma$ for $ 0 \leq k <m$.
 By choosing the step size $\beta$ on the input side and the step size $\gamma$ in the output side, we fix the FRFT parameter $\delta=\frac{\beta\gamma}{2\pi}$ and yield \cite{nzokem2021fitting} the density function $f$ (\ref{eq:l3}) at $x_k$.
 
  \begin{align}
 \hat{f}(x_{k+s}) = \frac{\gamma} {2\pi}e^{-\pi i(k+s-\frac{n}{2})n\delta}G_{k+s}(\scrF[f](y_{j})e^{-\pi i jn\delta}),-\delta) \hspace{5mm}
   \hbox{ $0\leq s <1$}
 \label {eq:l7}
 \end{align}  
\noindent
The numerical integration of functions is another method to evaluate the inverse Fourier integrals called the Direct Integration Method. One of the sophisticated procedures is the Newton–Cotes quadrature rules, where the interval is approximated by some interpolating polynomials, usually in Lagrange form. The 12-point rule Composite Newton-Cotes Quadrature can be implemented and provides greater accuracy. The error analysis shows that the global error $O(h^{13})$\cite{aubain2020, nzokem_2021,nzokem2021sis}. \\
\noindent
We assume $Q=12$ and $m=Qn$; $\beta=\frac{a}{m} $ is the step size of the $m$ input values $\scrF[f](y)$, defined by $y_{j+Qp}=(Qp+ j - \frac{m}{2}) \beta$ for $0 \leq p <n$ and $0 \leq j <Q$. Similarly, the output values of $f(x)$ is defined by $x_{Ql+f+s}=(Ql+f+s-\frac{m}{2}) \gamma$ for $0 \leq l <n$, $0 \leq f <Q$ and $0\leq s\leq 1$.

\begin{align}
 \hat{f}(x_{Ql+f+s}) = \beta \sum_{p=0}^{n-1} \sum_{j=0}^{Q} W_{j} e^{(iy_{j+Qp} x_{f+Ql+s})}F[f](y_{j+Qp}) \label {eq:l71} 
\end{align}
\noindent
The weight value $\{W_j\}_{0\leq j \leq Q}$ and details regarding the methodology are developed in \cite{aubain2020, nzokem_2021,nzokem2021sis}.

\subsection{ Advanced Fast Fourier Transform (FRFT) Algorithm}
\noindent
The Advanced Fast Fourier Transform (FRFT) algorithm combines the Fast Fractional Fourier (FRFT) algorithm (\ref{eq:l7}) and the 12-point rule Composite Newton-Cotes Quadrature (\ref{eq:l71}) to evaluate the inverse Fourier integrals.\\

\noindent
We assume that $\scrF[f](x)$ is zero outside the interval $[-\frac{a}{2}, \frac{a}{2}]$, $Q=12$, $m=Qn$ and $\beta=\frac{a}{m} $ is the step size of the $m$ input values $\scrF[f](y)$, defined by $y_{j+Qp}=(Qp+ j - \frac{m}{2}) \beta$ for $0 \leq p <n$ and $0 \leq j <Q$. Similarly, the output values of $f(x)$ is defined by $x_{Ql+f+s}=(Ql+f+s-\frac{m}{2}) \gamma$ for $0 \leq l <n$, $0 \leq f <Q$ and $0\leq s\leq 1$.

 \begin{equation} 
 \begin{aligned}
f(x_{Ql+f+s})&=\frac{1}{2\pi}\int_{\infty}^{+\infty}e^{i y x_{Ql+f+s}} F[f](y)dy = \frac{1}{2\pi}\int_{-a/2}^{a/2}e^{i y x_{Ql+f+s}} F[f](y)dy\\
 &= \frac{1}{2\pi}\sum_{p=0}^{n-1} \int_{y_{Qp}}^{y_{Qp + Q}}e^{i y x_{Ql+f+s}} F[f](y)dy \hspace{2mm} \hbox{ (composite rule)} \label {eq:l01}
\end{aligned}
\end{equation}

\noindent
Based on the Lagrange interpolating integration over [$x_{Qp}$, $x_{Qp+Q}$]\cite{aubain2020, nzokem_2021}, we have the following expression.
\begin{equation} 
\int_{x_{Qp}}^{x_{Qp + Q}}{e^{i y x_{Ql+f+s}}F[f](x)}dx  \approx \beta\sum_{j=0}^{Q} w_{j} {e^{i y x_{Ql+f+s}} F[f](x_{j + Qp})} \\
 \end{equation}
 
\noindent
we consider $\widehat{f}(x_{Ql+f+s})$ the approximation of $f(x_{Ql+f+s})$. It results that the expression (\ref{eq:l01}) becomes
 \begin{equation} 
 \begin{aligned}
\widehat{f}(x_{Ql+f+s}) &=\frac{\beta}{2\pi}\sum_{p=0}^{n-1} \sum_{j=0}^{Q} w_{j} F[f](y_{j + Qp})e^{i x_{Ql+f+s} y_{j + Qp}}\\
&= \frac{\beta}{2\pi}\sum_{j=0}^{Q}w_{j} \sum_{p=0}^{n-1}{F[f](y_{j + Qp})e^{2\pi i\delta (Ql+f+s-\frac{m}{2}) (Qp+ j - \frac{m}{2})}} \hspace{5mm}
   \hbox{ $\beta \gamma=2\pi\delta$}\\
&= \frac{\beta}{2\pi}e^{-\pi i\delta m(Ql+f+s-\frac{m}{2})}\sum_{j=0}^{Q}w_{j} G_{l+\frac{f+s}{Q}}(\xi_{p},\delta Q^{2})e^{2\pi i\delta (Ql-\frac{m}{2})j}e^{2\pi i\delta (f+s)j}
\end{aligned}
\end{equation}
\noindent
We have the first fast fractional Fourier transform (FRFT) to perform on $\{\xi_{p}\}_{0\leq p < n}$ 
 \begin{equation} 
 \begin{aligned}
G_{l+\frac{f+s}{Q}}(\xi_{p},-\delta Q^{2})&=\sum_{p=0}^{n-1}{y_{p}e^{-2\pi  i (l+\frac{f+s}{Q})p\delta Q^{2}}} \quad \quad \xi_{p}=e^{-\pi i m p Q \delta}F[f](y_{j + Qp})
\end{aligned}
\end{equation}

 \begin{equation} 
 \begin{aligned}
\widehat{f}(x_{Ql+f+s}) &= \frac{\beta}{2\pi}e^{-\pi i\delta m(Ql+f+s-\frac{m}{2})}\sum_{j=0}^{Q}w_{j} G_{l+\frac{f+s}{Q}}(\xi_{p},\delta Q^{2})e^{2\pi i\delta (Ql-\frac{m}{2})j}e^{2\pi i\delta (f+s)j}\\
&= \frac{\beta}{2\pi}e^{-\pi i\delta m(Ql+f+s-\frac{m}{2})}G_{f+s}(z_{j},\delta)
\end{aligned}
\end{equation}
\noindent
We have the second fast fractional Fourier transform (FRFT) to perform on $\{z_{j}\}_{0\leq j \leq Q}$ 
 \begin{equation} 
 \begin{aligned}
G_{f+s}(z_{j},-\delta)&=\sum_{j=0}^{Q}z_{j}e^{-2\pi i\delta (f+s)j} \quad \quad z_{j}=w_{j} G_{l+\frac{f+s}{Q}}(y_{p},\delta Q^{2})e^{2\pi i\delta (Ql-\frac{m}{2})j}
\end{aligned}
\end{equation}
The advanced FRFT - scheme yields the following approximation
 \begin{equation} 
\widehat{f}(x_{Ql+f+s}) = \frac{\beta}{2\pi}e^{-\pi i\delta m(Ql+f+s-\frac{m}{2})}G_{f+s}(z_{j},\delta)
\end{equation}
\section{Fitting General Tempered Stable Distribution Results}
\subsection{Review of the Maximum Likelihood Method} \label{sect1}
\noindent
From a probability density function $f(x, V)$ with parameter $V$ of size $p=7$ and the sample data $X$ of size $m$,  we definite  the Likelihood function and its derivatives as follows: 
\begin{align}
 L(x,V) &= \prod_{j=1}^{m} f(x_{j},V) \quad & 
 l(x,V) &= \sum_{j=1}^{m} log(f(x_{j},V))  \label {eq:l32} 
  \end{align}
 
\begin{align}
 \frac{dl(x,V)}{dV_j} &= \sum_{i=1}^{m} \frac{\frac{df(x_{i},V)}{dV_j}}{f(x_{i},V)} \quad &
  \frac{d^{2}l(x,V)}{dV_{k}dV_{j}} &= \sum_{i=1}^{m} \left(\frac{\frac{d^{2}f(x_{i},V)}{dV_{k}dV_{j}}}{f(x_{i},V)}- \frac{\frac{df(x_{i},V)}{dV_{k}}}{f(x_{i},V)}\frac{\frac{df(x_{i},V)}{dV_j}}{f(x_{i},V)}\right) \label {eq:l35}
 \end{align}
\noindent
To perform the Maximum of the likelihood function (\ref{eq:l32}), the quantities  $\frac{df(x, V)}{dV_j}$ and $\frac{d^{2}f(x, V)}{dV_{k}dV_{j}}$ in (\ref{eq:l35}) are the first and second order derivatives of the probability density and should be computed with high accuracy.
\noindent
The quantities $\frac{d^{2}l(x, V)}{dV_{k}dV_{j}}$ are critical in computing the Hessian Matrix and the Fisher Information Matrix.\\
\noindent
Given the parameter $V=(\textbf{$\mu$}, \textbf{$\beta_{+}$}, \textbf{$\beta_{-}$}, \textbf{$\alpha_{+}$},\textbf{$\alpha_{-}$}, \textbf{$\lambda_{+}$}, \textbf{$\lambda_{-}$})$ and the sample data set $X$, we derive from (\ref{eq:l35}) the following vector and matrix (\ref{eq:l36}) .
  \begin{align}
 I'(X,V) =\left(\frac{dl(x,V)}{dV_j}\right)_{1 \leq j \leq p}  \quad  \quad  I''(X,V) = \left(\frac{d^{2}l(x,V)}{dV_{k}dV_{j}}\right)_{\substack{{1 \leq k \leq p} \\ {1 \leq j \leq p}}} \label {eq:l36}
 \end{align}
\noindent
The advanced FRFT - scheme developed previously is used to compute the likelihood function (\ref{eq:l32}) and its derivatives (\ref{eq:l36}) in the optimization process. \\

\noindent
The local solution $V^{0}$ should meet the following requirements.
  \begin{align}
 I'(x,V^{0})=0 \quad \quad U^{T}\mathbf{I''(X,V^{0})}U \leq 0\hspace{5mm}  \hbox{ $\forall U \in \mathbb{R}^{p}$}\label{eq:l37}
 \end{align}
The inequality in (\ref{eq:l37}) is met when $I''(x, V^{0})$ is a negative semi-definite matrix.\\
\noindent
The Newton-Raphson Iteration Algorithm provides the numerical solution (\ref{eq:l38}). 
  \begin{align}
V^{n+1}=V^{n}-{\left(I''(x,V^{n})\right)^{-1}}I'(x,V^{n})\label{eq:l38}
 \end{align}
\noindent
See \cite{giudici2013wiley} for details on Maximum likelihood and Newton-Raphson Iteration procedure.
\newpage
\subsection{GTS Parameter Estimation from S$\&$P 500 index}
\noindent 
The results of the GTS Parameter Estimation from S$\&$P 500 return data are reported in Table \ref{tab1}. As expected, the stability indexes (\textbf{$\beta_{-}$},\textbf{$\beta_{+}$}), the process intensities (\textbf{$\alpha_{-}$},\textbf{$\alpha_{+}$}), and the decay rate (\textbf{$\lambda_{-}$},\textbf{$\lambda_{+}$}) are all positive. In addition, the results show $0\le \beta_{+} \le 1$ and $0\le \beta_{-} \le 1$; therefore, the positive S$\&$P 500 rate of return ($X_{+}$) and the negative S$\&$P 500 rate of return ($X_{-}$) are infinite activity processes, which means processes with an infinite number of jumps in any given time interval. 
\begin{table}[ht]
\vspace{-0.3cm}
\caption{ FRFT Maximum Likelihood GTS Parameter Estimation}
\label{tab1} 
\vspace{-0.3cm}
\centering
\begin{tabular}{@{}lccccccc@{}}
\toprule
\textbf{Model} & \textbf{$\mu$} & \textbf{$\beta_{+}$} & \textbf{$\beta_{-}$} & \textbf{$\alpha_{+}$} & \textbf{$\alpha_{-}$}  & \textbf{$\lambda_{+}$}  & \textbf{$\lambda_{-}$}  \\ \midrule
\textbf{GTS} & -0.693477 & 0.682290 & 0.242579 & 0.458582 & 0.414443 & 0.822222 & 0.727607  \\ \bottomrule 
\end{tabular}
\vspace{-0.2cm}
\end{table} 

\noindent
As shown in Table \ref{tab1}, We have \textbf{$\alpha_{-}$} $\leq$ \textbf{$\alpha_{+}$} and the higher arrival rate of jump on the right side contributes to the increasing global trend of the S$\&$P 500 daily price in Fig \ref{fig02}. Regarding the skewness parameters, we have \textbf{$\lambda_{-}$} $\leq$ \textbf{$\lambda_{+}$} and the S$\&$P 500 return is a bit left-skewed distribution. The parameter analysis shows that the tail distribution is thicker on the negative side of the S$\&$P 500 return distribution ($X_{-}$) than on the positive side of the distribution.\\

\noindent
 The implementation of the maximum likelihood method (in \ref{sect1}) was made possible by the numerical computations of the GTS probability density function $f(x, V)$ and its first and second derivative functions by the advanced FRFT - scheme developed previously in section 4. The results of the iteration process are reported in Table \ref{tab2}. each raw has eleven columns made of the iteration number, the seven parameters \textbf{$\mu$}, \textbf{$\beta_{+}$}, \textbf{$\beta_{-}$}, \textbf{$\alpha_{+}$},\textbf{$\alpha_{-}$}, \textbf{$\lambda_{+}$}, \textbf{$\lambda_{-}$} and three statistical indicators \textbf{$Log(ML)$}, \textbf{$||\frac{dLog(ML)}{dV}||$}, \textbf{$Max Eigen Value$}. The statistical indicators aim at checking if the two necessary and sufficient conditions described in equations(\ref{eq:l37}) are all met. For each iteration in Table \ref{tab2}, \textbf{$Log(ML)$} displays the value of the Naperian logarithm of the likelihood function $L(x, V)$ as described in (\ref{eq:l32}); \textbf{$||\frac{dLog(ML)}{dV}||$} displays the value of the norm of the first derivatives (\textbf{$\frac{dl(x, V)}{dV_j}$}) described in equations (\ref{eq:l35}); and \textbf{$Max Eigen Value$} displays the maximum value of the seven Eigenvalues generated by the Hessian Matrix (\textbf{$\frac{d^{2}l(x, V)}{dV_{k}dV_{j}}$}) as described in equations (\ref{eq:l36}).\\
 
\noindent
The Newton-Raphson Iteration Algorithm (\ref{eq:l38}) was implemented, and the results are reported in Table \ref{tab2}. As shown in Table \ref{tab2}, the log-likelihood (\textbf{$Log(ML)$}) value starts at $-4664.7647$ and increases to a limit of $-4659.1914$; the \textbf{$||\frac{dLog(ML)}{dV}||$} value starts at $289.206866$ to decreases to $0$; and the maximum value of the Eigenvalues (\textbf{$Max Eigen Value$}) start at $48.3287566$ and converge to $-1.4542632$, which is negative. At the convergent solution, the Hessian matrix in (\ref{eq:l36}) is a negative semi-definite matrix; both conditions in (\ref{eq:l37}) are met, and we have a locally optimal solution.
\begin{table}[ht]
\vspace{-0.3cm}
\centering
\caption{GTS Parameter Estimations from S$\&$P 500 index}
\label{tab2}
\vspace{-0.3cm}
\resizebox{13cm}{!}{%
\begin{tabular}{ccccccccccc}
\hline
\textbf{$Iterations$} & \textbf{$\mu$} & \textbf{$\beta_{+}$} & \textbf{$\beta_{-}$} & \textbf{$\alpha_{+}$} & \textbf{$\alpha_{-}$} & \textbf{$\lambda_{+}$} & \textbf{$\lambda_{-}$} & \textbf{$Log(ML)$} & \textbf{$||\frac{dLog(ML)}{dV}||$} & \textbf{$Max Eigen Value$} \\ \hline
1  & -0.5266543 & 0.67666185 & 0.43662658 & 0.43109407 & 0.32587754 & 0.85012595 & 0.60145985 & -4664.7647 & 289.206866 & 48.3287566 \\
2  & -0.5459715 & 0.67059496 & 0.42415971 & 0.44989796 & 0.34810329 & 0.80921924 & 0.60289955 & -4660.2156 & 35.9853332 & -6.0431046 \\
3  & -0.7108484 & 0.66676659 & 0.20613817 & 0.46963149 & 0.41246296 & 0.8368645  & 0.73679498 & -4663.5777 & 1082.00304 & 449.252496 \\
4  & -0.6704327 & 0.66891111 & 0.11250993 & 0.46528242 & 0.50028087 & 0.83422813 & 0.86330646 & -4660.5268 & 135.684998 & 15.1103621 \\
5  & -0.739816  & 0.66779246 & 0.09097181 & 0.48302677 & 0.45922815 & 0.85551132 & 0.81320574 & -4660.0208 & 45.7082317 & 10.8419307 \\
6  & -0.6517205 & 0.65580459 & 0.19673172 & 0.47996977 & 0.42743878 & 0.85248358 & 0.75350741 & -4659.833  & 46.3329967 & 11.8534517 \\
7  & -0.8136505 & 0.71997129 & 0.21558026 & 0.44023152 & 0.41952419 & 0.79423596 & 0.74025892 & -4662.4815 & 1187.98211 & 166.006596 \\
8  & -0.7805857 & 0.70635344 & 0.22952102 & 0.44666893 & 0.4166063  & 0.80360771 & 0.7334486  & -4659.7758 & 85.6584355 & -3.4305336 \\
9  & -0.7543747 & 0.69912144 & 0.23468747 & 0.45029318 & 0.41542774 & 0.80935861 & 0.73077041 & -4659.1944 & 1.07446211 & -0.7993855 \\
10 & -0.7533784 & 0.69885561 & 0.23480071 & 0.45042677 & 0.41541378 & 0.80956671 & 0.73072468 & -4659.1943 & 1.03706137 & -0.8136858 \\
11 & -0.752414  & 0.69859801 & 0.2349108  & 0.45055614 & 0.41540021 & 0.8097682  & 0.73068024 & -4659.1942 & 1.00184387 & -0.8273569 \\
12 & -0.7514794 & 0.69834808 & 0.23501796 & 0.45068156 & 0.41538701 & 0.80996353 & 0.73063702 & -4659.1942 & 0.96860649 & -0.8404525 \\
13 & -0.7496917 & 0.69786927 & 0.23522417 & 0.4509216  & 0.41536164 & 0.81033731 & 0.73055392 & -4659.194  & 0.90739267 & -0.8650984 \\
14 & -0.7471903 & 0.69719768 & 0.23551545 & 0.45125776 & 0.41532584 & 0.81086063 & 0.73043671 & -4659.1938 & 0.82669598 & -0.8987475 \\
15 & -0.7463993 & 0.69698491 & 0.23560822 & 0.45136413 & 0.41531445 & 0.8110262  & 0.73039942 & -4659.1937 & 0.80232744 & -0.9091961 \\
16 & -0.7456279 & 0.69677723 & 0.23569899 & 0.45146791 & 0.41530331 & 0.8111877  & 0.73036294 & -4659.1937 & 0.77907591 & -0.9193009 \\
17 & -0.7434222 & 0.69618233 & 0.23596025 & 0.45176484 & 0.41527126 & 0.81164974 & 0.73025805 & -4659.1935 & 0.71530396 & -0.9477492 \\
18 & -0.7427203 & 0.69599272 & 0.2360439  & 0.45185938 & 0.41526101 & 0.81179683 & 0.73022449 & -4659.1934 & 0.69583386 & -0.9566667 \\
19 & -0.7376152 & 0.69460887 & 0.23665992 & 0.45254797 & 0.41518555 & 0.81286777 & 0.72997772 & -4659.193  & 0.56549166 & -1.0197019 \\
20 & -0.7353516 & 0.69399264 & 0.23693733 & 0.45285383 & 0.41515159 & 0.81334323 & 0.72986678 & -4659.1929 & 0.51375646 & -1.0466686 \\
21 & -0.7342715 & 0.69369811 & 0.23707047 & 0.45299987 & 0.41513529 & 0.81357023 & 0.72981356 & -4659.1928 & 0.47702329 & -1.0633442 \\
22 & -0.727748  & 0.69191195 & 0.23788938 & 0.45388257 & 0.41503509 & 0.81494168 & 0.72948683 & -4659.1924 & 0.30254549 & -1.1564316 \\
23 & -0.7269375 & 0.69168915 & 0.23799284 & 0.45399234 & 0.41502243 & 0.81511217 & 0.72944561 & -4659.1924 & 0.28659447 & -1.1671214 \\
24 & -0.7261524 & 0.69147312 & 0.23809343 & 0.45409871 & 0.41501011 & 0.81527736 & 0.72940554 & -4659.1923 & 0.27235218 & -1.177299  \\
25 & -0.7074205 & 0.68625723 & 0.24059125 & 0.45665159 & 0.41470173 & 0.81923664 & 0.72841073 & -4659.1916 & 0.15127988 & -1.3705213 \\
26 & -0.7028887 & 0.68497607 & 0.24122265 & 0.45727561 & 0.4146222  & 0.82020261 & 0.72815815 & -4659.1915 & 0.12096032 & -1.4033087 \\
27 & -0.6987925 & 0.68381109 & 0.24180316 & 0.45784257 & 0.4145479  & 0.82107955 & 0.72792474 & -4659.1914 & 0.07862954 & -1.4283679 \\
28 & -0.6934505 & 0.68228736 & 0.24256739 & 0.45858391 & 0.41444901 & 0.82222574 & 0.72761631 & -4659.1914 & 0.75324649 & -1.6442249 \\
29 & -0.6934765 & 0.68229002 & 0.24257963 & 0.45858236 & 0.41444404 & 0.82222285 & 0.72760762 & -4659.1914 & 9.93E-06   & -1.4542674 \\
30 & -0.6934774 & 0.68229032 & 0.24257975 & 0.45858219 & 0.41444396 & 0.8222226  & 0.7276075  & -4659.1914 & 7.4193E-08 & -1.4542632\\ \hline
\end{tabular}%
}
\vspace{-0.3cm}
\end{table}

\newpage
\noindent
GTS Parameter Estimation (in Table \ref{tab1}) was used to compare the empirical values and theoretical values of the first derivative of the GTS probability density function. The comparison provides the degree of accuracy for the advanced FRFT scheme and the fitness between empirical and theoretical values. As shown in Fig \ref{fig:5}, the empirical values $\left(\frac{\frac{df(x_{i}, V)}{dV_j}}{f(x_{i}, V)}\right)_{\substack{{1 \leq i \leq m} \\ {1 \leq j \leq 7}}}$ generated by S$\&$P 500 returns are displayed in circles (red for the parameter positive sign and black for the parameter negative sign); and the theoretical values $\left(\frac{\frac{df(x, V)}{dV_j}}{f(x, V)}\right)_{\substack{{1 \leq j \leq 7}}}$ are displayed in a straight curve. As shown in Fig \ref{fig:05} and Fig \ref{fig:07}, \textbf{$\beta_{+}$} and \textbf{$\alpha_{+}$} have a higher effect on the Probability density function (pdf) than \textbf{$\beta_{-}$} and \textbf{$\alpha_{-}$} respectively. However, in Fig \ref{fig:06}, \textbf{$\lambda_{-}$} and \textbf{$\lambda_{+}$} are symmetric and have the same effect on pdf.
 \begin{figure}[ht]
\vspace{-0.3cm}
     \centering
     \begin{subfigure}[b]{0.33\textwidth}
         \centering
         \includegraphics[width=\textwidth]{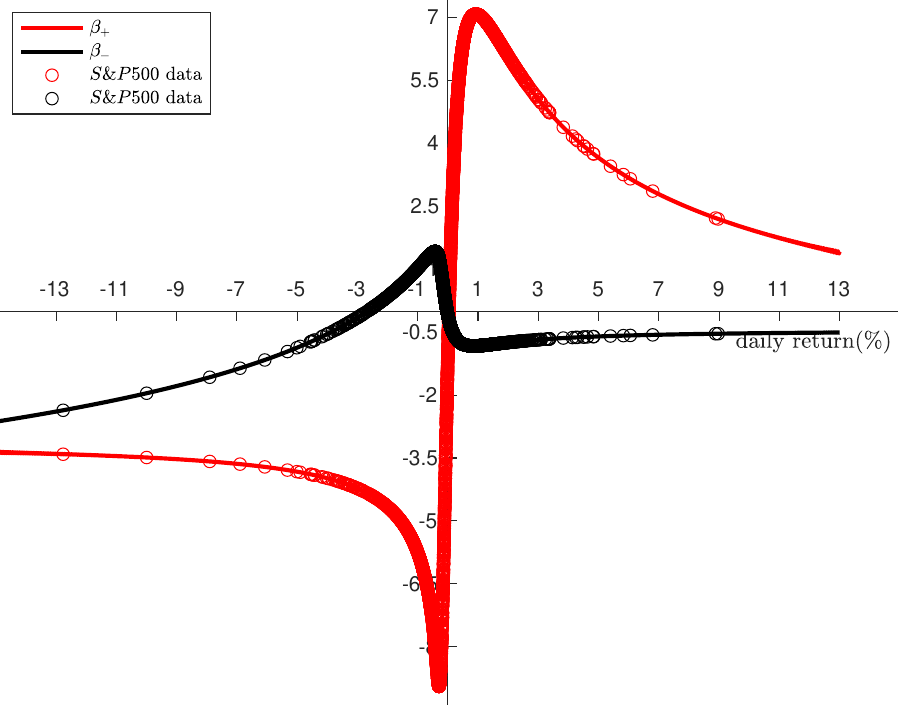}
         \vspace{-0.6cm}
         \caption{$V_j=\beta^{+}$, $V_j=\beta^{-}$}
         \label{fig:05}
     \end{subfigure}
     \begin{subfigure}[b]{0.32\textwidth}
         \centering
         \includegraphics[width=\textwidth]{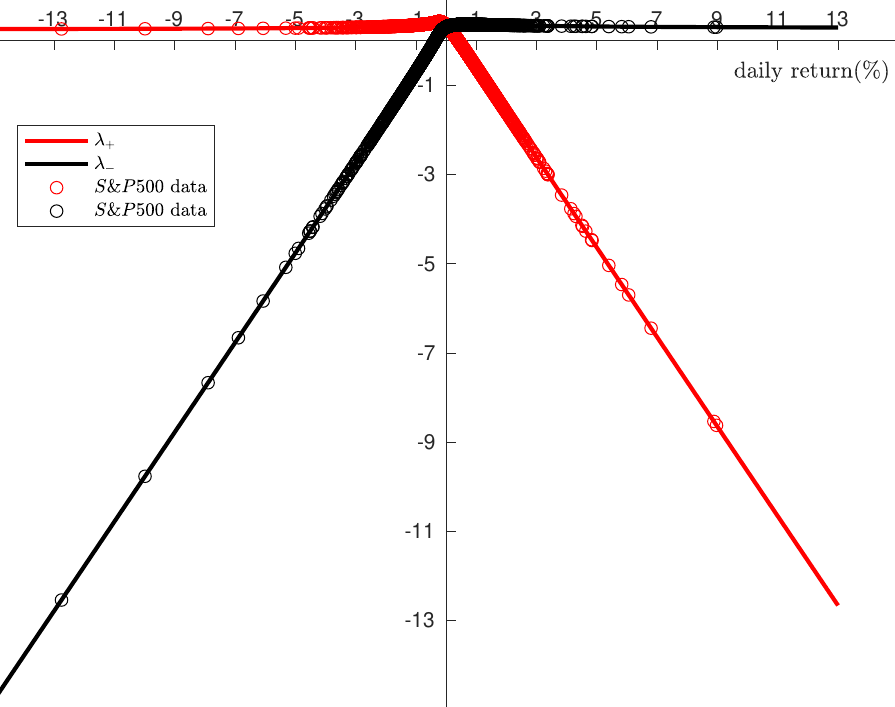}
         \vspace{-0.6cm}
         \caption{$V_j=\lambda^{+}$, $V_j=\lambda^{-}$}
         \label{fig:06}
     \end{subfigure}
       \begin{subfigure}[b]{0.33\textwidth}
         \centering
         \includegraphics[width=\textwidth]{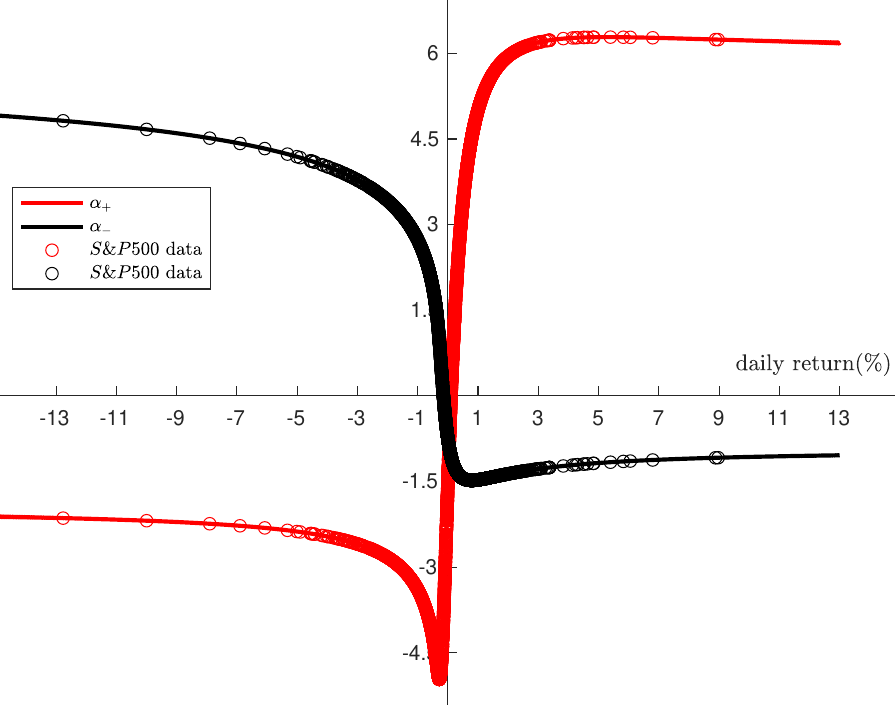}
         \vspace{-0.6cm}
         \caption{$V_j=\alpha^{+}$, $V_j=\alpha^{-}$}
         \label{fig:07}
     \end{subfigure}
        \caption{ $\frac{\frac{df(x_{i},V)}{dV_j}}{f(x_{i},V)}$}
        \label{fig:5}
        \vspace{-0.3cm}
\end{figure}
\newpage

\subsection{GTS Parameter Estimation: Bitcoin}
\noindent 
The results of the GTS Parameter Estimation from Bitcoin returns are reported in Table \ref{tab3}. As expected, the stability indexes (\textbf{$\beta_{-}$},\textbf{$\beta_{+}$}), the process intensities (\textbf{$\alpha_{-}$},\textbf{$\alpha_{+}$}), and the decay rate (\textbf{$\lambda_{-}$},\textbf{$\lambda_{+}$}) are all positive. The results show $0\le \beta_{+} \le 1$ and $0\le \beta_{-} \le 1$, the positive Bitcoin return ($X_{+}$) and the negative Bitcoin return ($X_{-}$) are infinite activity processes, which means each process has an infinite number of jumps in any given time interval. 
\begin{table}[ht]
\caption{ FRFT Maximum Likelihood GTS Parameter Estimation}
\label{tab3} 
\vspace{-0.3cm}
\centering
\begin{tabular}{@{}lccccccc@{}}
\toprule
\textbf{Model} & \textbf{$\mu$} & \textbf{$\beta_{+}$} & \textbf{$\beta_{-}$} & \textbf{$\alpha_{+}$} & \textbf{$\alpha_{-}$}  & \textbf{$\lambda_{+}$}  & \textbf{$\lambda_{-}$}  \\ \midrule
\textbf{GTS}  & -0.736924  & 0.461378 & 0.267178 & 0.810017 & 0.517347 & 0.215628 & 0.191937   \\ \bottomrule 
\end{tabular}
\vspace{-0.2cm}
\end{table} 

\noindent
 As shown in Table \ref{tab1}, We have \textbf{$\alpha_{-}$} $\leq$ \textbf{$\alpha_{+}$} and the higher arrival rate of jump on the right side ($X_{+}$) contributes to the increasing global trend of the bitcoin daily price in Fig \ref{fig02}. Regarding the skewness parameters, we have \textbf{$\lambda_{-}$} $\leq$ \textbf{$\lambda_{+}$} and the S$\&$P 500 return is a bit left-skewed distribution. The parameter analysis shows that the tail distribution is thicker on the negative side of the Bitcoin return distribution ($X_{-}$) than on the positive side.\\
 
\noindent
The Newton-Raphson Iteration Algorithm (\ref{eq:l38}) was implemented, and the results are reported in Table \ref{tab4}. As shown in Table \ref{tab4}, the log-likelihood (\textbf{$Log(ML)$}) value starts at $-10036.656$ and increases to a limit of $-9751.2193$; the \textbf{$||\frac{dLog(ML)}{dV}||$} value starts at $3188.1469$ and decreases to almost $0$; and the maximum value of the Eigenvalues (\textbf{$Max Eigen Value$}) start at $2854.2231$ and converge to $-0.000436$, which is negative. At the convergent solution, the Hessian matrix in (\ref{eq:l36}) is a negative semi-definite matrix; both conditions in (\ref{eq:l37}) are met, and we have a locally optimal solution.
\begin{table}[ht]
\vspace{-0.3cm}
\centering
\caption{GTS Parameter Estimations from Bitcoin returns}
\label{tab4}
\vspace{-0.3cm}
\resizebox{13cm}{!}{%
\begin{tabular}{ccccccccccc}
\hline
\textbf{$Iterations$} & \textbf{$\mu$} & \textbf{$\beta_{+}$} & \textbf{$\beta_{-}$} & \textbf{$\alpha_{+}$} & \textbf{$\alpha_{-}$} & \textbf{$\lambda_{+}$} & \textbf{$\lambda_{-}$} & \textbf{$Log(ML)$} & \textbf{$||\frac{dLog(ML)}{dV}||$} & \textbf{$Max Eigen Value$} \\ \hline
1 & -0.5973 & 0.3837 & 0.4206 & 1.872 & 0.694 & 0.5333 & 0.219 & -10036.656 & 3188.1469 & 2854.2231 \\
2 & -1.2052311 & 0.35850753 & 0.44451399 & 2.00512495 & 0.62313779 & 0.53840744 & 0.20635235 & -9912.1541 & 1971.95236 & 1518.15345 \\
3 & -1.624469 & 0.31755508 & 0.45798507 & 2.23748588 & 0.58475918 & 0.56388213 & 0.17998332 & -9865.5242 & 952.027546 & 716.347539 \\
4 & -1.9959515 & 0.2615866 & 0.47359748 & 2.57864607 & 0.5494987 & 0.60990478 & 0.16377681 & -9862.2551 & 440.489423 & 295.233361 \\
5 & -2.8340545 & -0.0033236 & 0.55268657 & 4.51611949 & 0.47359381 & 0.87264266 & 0.130233 & -9911.337 & 510.396179 & 60.6518499 \\
6 & -2.701598 & 0.07378625 & 0.54127007 & 3.76512451 & 0.47819392 & 0.7743675 & 0.13294829 & -9896.0802 & 611.613266 & 117.692562 \\
7 & -1.4130144 & 0.37715064 & 0.48510637 & 1.50024644 & 0.5523787 & 0.39998447 & 0.15474644 & -9804.2548 & 801.962681 & 307.039536 \\
8 & -1.049984 & 0.45100489 & 0.4741659 & 1.13318569 & 0.56070768 & 0.30817133 & 0.15667085 & -9775.871 & 728.753228 & 295.768229 \\
9 & -0.8211241 & 0.50542632 & 0.46591368 & 0.90824445 & 0.54487666 & 0.23555274 & 0.1533588 & -9756.7912 & 397.102674 & 133.652433 \\
10 & -0.7319208 & 0.53998085 & 0.44372538 & 0.81377008 & 0.54717214 & 0.19930192 & 0.15902826 & -9752.0888 & 89.8942291 & 12.4808028 \\
11 & -1.2317854 & 0.58475806 & 0.23780021 & 0.81260395 & 0.49719481 & 0.18597445 & 0.19343462 & -9751.3131 & 32.1533364 & -2.5188051 \\
12 & -0.7150133 & 0.4721273 & 0.29761245 & 0.80521805 & 0.52356611 & 0.21204916 & 0.18644449 & -9750.6869 & 5.52013773 & 0.11490414 \\
13 & -0.7271098 & 0.45992936 & 0.27054512 & 0.80932883 & 0.51829775 & 0.21585683 & 0.19139837 & -9751.2179 & 1.13726334 & 0.19172007 \\
14 & -0.7750442 & 0.46778988 & 0.25646086 & 0.81245021 & 0.51411912 & 0.21450942 & 0.19360685 & -9751.2252 & 1.00910053 & -0.9423667 \\
15 & -0.7090246 & 0.45678856 & 0.27535855 & 0.80819689 & 0.51977265 & 0.21641582 & 0.19065445 & -9751.2151 & 0.68576109 & 0.57268124 \\
16 & -0.7494266 & 0.46338031 & 0.26341052 & 0.81084098 & 0.51624184 & 0.21529106 & 0.19253194 & -9751.2211 & 0.65747338 & -0.2888343 \\
17 & -0.7429071 & 0.46233754 & 0.26537785 & 0.8104111 & 0.51681851 & 0.21546689 & 0.19222166 & -9751.2201 & 0.65058847 & -0.1361475 \\
18 & -0.7401206 & 0.46189156 & 0.2662183 & 0.81022749 & 0.51706495 & 0.21554213 & 0.19208913 & -9751.2197 & 0.64914609 & -0.0724136 \\
19 & -0.7369888 & 0.4613886 & 0.26715947 & 0.81002148 & 0.51734136 & 0.21562718 & 0.19194083 & -9751.2193 & 0.64887078 & -0.0018685 \\
20 & -0.7369633 & 0.4613845 & 0.26716714 & 0.8100198 & 0.51734361 & 0.21562787 & 0.19193962 & -9751.2193 & 0.64887314 & -0.0012986 \\
21 & -0.7369455 & 0.46138165 & 0.26717247 & 0.81001863 & 0.51734517 & 0.21562835 & 0.19193878 & -9751.2193 & 0.64887482 & -0.0009025 \\
22 & -0.7369332 & 0.46137967 & 0.26717617 & 0.81001782 & 0.51734626 & 0.21562869 & 0.1919382 & -9751.2193 & 0.64887601 & -0.0006273 \\
23 & -0.7369246 & 0.46137829 & 0.26717875 & 0.81001726 & 0.51734702 & 0.21562892 & 0.19193779 & -9751.2193 & 0.64887685 & -0.000436\\ \hline
\end{tabular}
}
\vspace{-0.2cm}
\end{table}

\noindent
GTS parameter estimation (in Table \ref{tab3}) was used to compare the empirical values and theoretical values of the first derivative of the GTS probability density function. Similar to Fig \ref{fig:5}, in Fig \ref{fig:9}, the empirical values $\left(\frac{\frac{df(x_{i}, V)}{dV_j}}{f(x_{i}, V)}\right)_{\substack{{1 \leq i \leq m} \\ {1 \leq j \leq 7}}}$ generated by the Bitcoin returns are displayed in circles (red for the parameter positive sign and black for the parameter negative sign; and the theoretical values $\left(\frac{\frac{df(x, V)}{dV_j}}{f(x, V)}\right)_{\substack{{1 \leq j \leq 7}}}$ are displayed in a straight curve. As shown in Fig \ref{fig:08}, \textbf{$\beta_{-}$} has a higher effect on the Probability density function (pdf) than \textbf{$\beta_{+}$}, but \textbf{$\alpha_{+}$} has a higher effect than \textbf{$\alpha_{-}$} in Fig \ref{fig:10}. However, in Fig \ref{fig:09}, \textbf{$\lambda_{-}$} and \textbf{$\lambda_{+}$} are symmetric and have the same effect on pdf.

\begin{figure}[ht]
     \centering
     \begin{subfigure}[b]{0.34\textwidth}
         \centering
         \includegraphics[width=\textwidth]{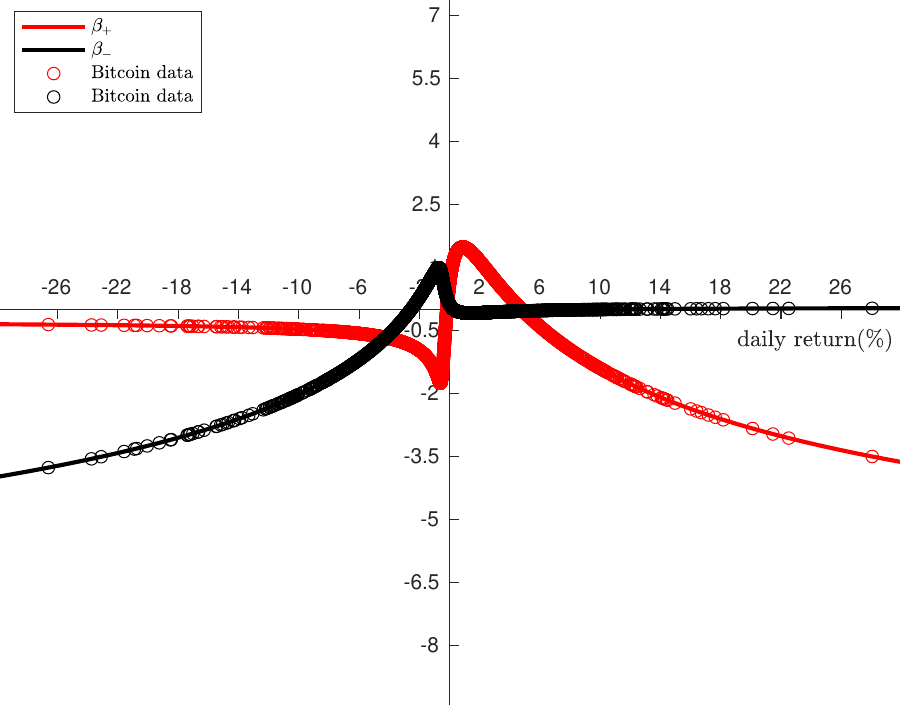}
         \caption{ $V_j=\beta^{+}$, $V_j=\beta^{-}$}
         \label{fig:08}
     \end{subfigure}
     \begin{subfigure}[b]{0.3\textwidth}
         \includegraphics[width=\textwidth]{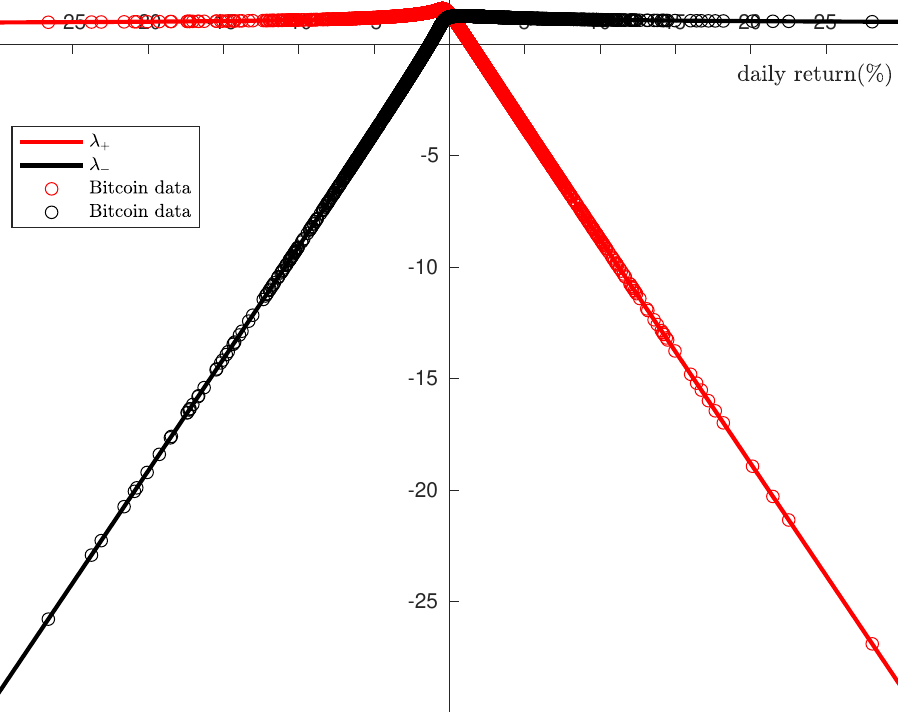}
         \caption{$V_j=\lambda^{+}$, $V_j=\lambda^{-}$}
         \label{fig:09}
     \end{subfigure}
       \begin{subfigure}[b]{0.34\textwidth}
         \includegraphics[width=\textwidth]{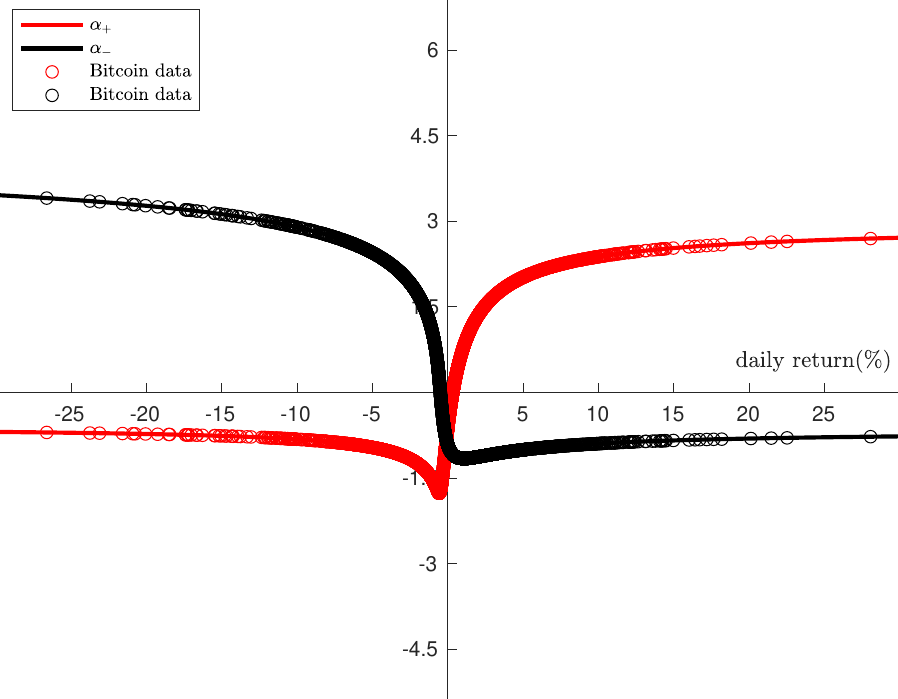}
         \caption{$V_j=\alpha^{+}$, $V_j=\alpha^{-}$}
         \label{fig:10}
     \end{subfigure}
        \caption{ $\frac{\frac{df(x_{i},V)}{dV_j}}{f(x_{i},V)}$}
        \label{fig:9}
        \vspace{-0.3cm}
\end{figure}
\noindent
The difference in amplitude effect is more visible when Fig \ref{fig:5} (S$\&$P 500 index) and Fig \ref{fig:10} (Bitcoin) are put together with each GST parameter side by side. As shown in the Appendix \ref{eq:an1}, the stability indexes (\textbf{$\beta_{-}$},\textbf{$\beta_{+}$}) and the process intensities (\textbf{$\alpha_{-}$},\textbf{$\alpha_{+}$}) degenerated from S$\&$P 500 index fitting have higher effects on the GST probability density function (pdf) whereas the decay rate (\textbf{$\lambda_{-}$},\textbf{$\lambda_{+}$}) degenerated from Bitcoin fitting has higher effects on the GST probability density function (pdf).


\section{Analysis and Findings: Density function and Key Statistics }

\noindent
The GTS parameter estimates from S$\&$P 500 return and Bitcoin return data are used to compute the GTS probability density functions, and their main characteristics are analyzed. As illustrated in Fig\ref{fig:14}, Both graphs show that tail events are much more prevalent than we would expect with a Normal distribution. The heavy-tailed distribution captures the huge price swings of the Bitcoin and S$\&$P 500 index. According to Table \ref{tab5}, the theoretical Kurtosis statistics are $8.92319$ for S$\&$P 500 returns and $9.74633$ for Bitcoin returns, which are almost three times the kurtosis of the Normal distribution. The peakedness of the density function is another characteristic, as shown in Fig\ref{fig:14}. In contrast to the Normal distribution, there is a higher concentration of data values around the mean. The degree of concentration is much higher for S$\&$P 500 return in Fig\ref{fig:12} than Bitcoin return distribution in Fig\ref{fig:11}. Many studies show that Kurtosis is not a measure of peakedness but rather a measure of tailedness\cite{horn1983measure, kochanski2022does, westfall2014kurtosis}. Both GTS distributions are called in the literature leptokurtic distributions.\\
\noindent
For Bitcoin return distributions in Fig\ref{fig:11}, the Normal probability density (red curve) with mean ($\kappa_{1}=0.15\%$) and standard deviation ($\sigma=3.99\%$) is plotted alongside the GTS probability density (black curve). According to the Normal distribution, about $35.15\%$ of Bitcoin returns ( in the purple area) are within $-1.57\%$ and $2.07\%$, $19.63\%$ lower than the actual percentage of Bitcoin returns. According to the GTS distribution, about $20.72\%$ of Bitcoin returns (in the blue area) are within $-9.53\%$ and $-1.57\%$; and on the right side, $20.99\%$ of Bitcoin returns (in the yellow area) are within $2.07\%$ and $10.05\%$. However, as shown in the red and blue areas in Fig\ref{fig:11}, both percentages are overestimated by the normal distribution.
 \begin{figure}[ht]
 \vspace{-0.3cm}
     \centering
     \begin{subfigure}[b]{0.48\textwidth}
         \centering
         \includegraphics[width=\textwidth]{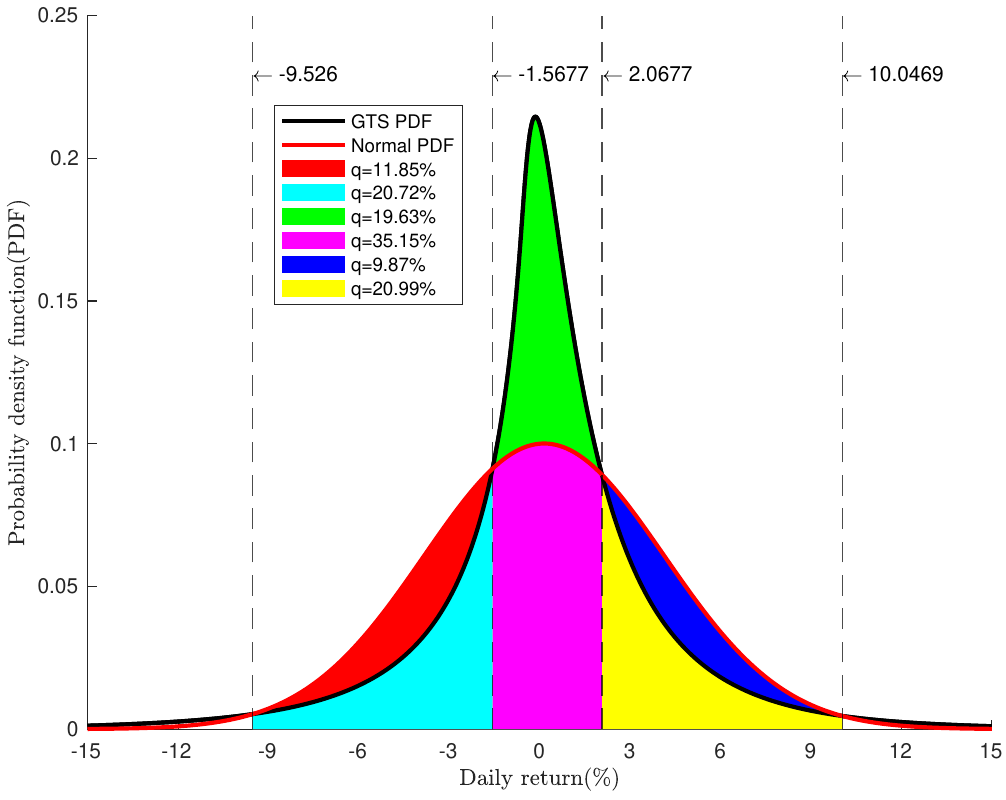}
         \vspace{-0.6cm}
         \caption{Bitcoin return distribution}
         \label{fig:11}
     \end{subfigure}
     \begin{subfigure}[b]{0.48\textwidth}
         \centering
         \includegraphics[width=\textwidth]{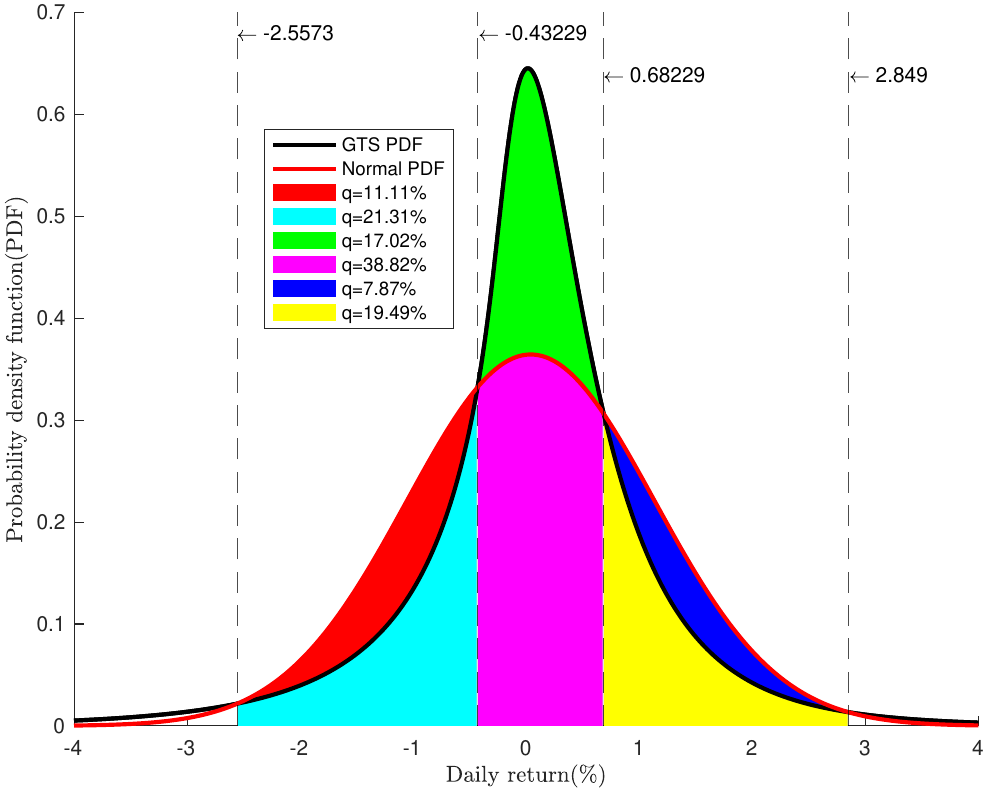}
         \vspace{-0.6cm}
         \caption{S$\&$P 500 return distribution}
         \label{fig:12}
     \end{subfigure}
     \vspace{-0.3cm}
        \caption{Probability density functions}
        \label{fig:14}
       \vspace{-0.3cm}
\end{figure}

\noindent
For S$\&$P 500 return distributions in Fig\ref{fig:12}, the Normal probability density (red curve) with a mean ($\kappa_{1}=0.04\%$) and a standard deviation ($\sigma=1.09\%$) is plotted along with the GTS probability density (black curve). Similar to the previous analysis, according to the Normal distribution, about $38.82\%$ of S$\&$P 500 index returns are concentrated within $-0.43\%$ and $0.68\%$; the percentage is $17.02\%$ lower than the actual percentage of S$\&$P 500 returns. According to the GTS distribution, about $21.31\%$ of S$\&$P 500 (in cyan area) are within $-2.56\%$ and $-0.43\%$; and on the right side, $19.49\%$ of S$\&$P 500 returns (in yellow area) are within $0.68\%$ and $2.85\%$. However, as shown in the red and blue areas in Fig\ref{fig:11}, both percentages are overestimated by the normal distribution.\\
\noindent
The GTS probability density functions from S$\&$P 500 returns and Bitcoin returns are compared in Fig\ref{fig:13}. The heavy-tailedness is the main characteristic of the GTS probability density function, whereas the GTS probability density function from the S$\&$P 500 index is characterized by peakedness.\\
 \begin{figure}[ht]
 \vspace{-0.6cm}
     \centering
         \includegraphics[width=0.55\textwidth]{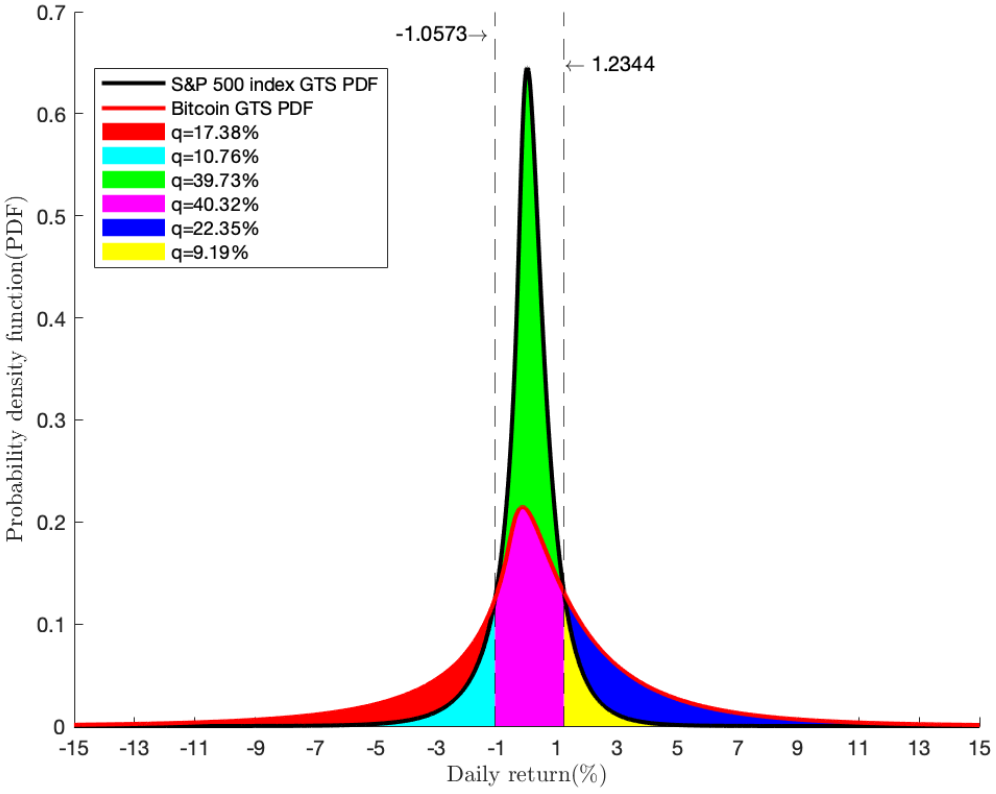}
         \vspace{-0.4cm}
         \caption{Bitcoin Versus $S\&P$ 500 returns}
         \label{fig:13}
         \vspace{-0.3cm}
\end{figure}

\newpage
\noindent
For Bitcoin return distribution in Fig\ref{fig:13}, about $40.32\%$ of Bitcoin returns ( in the purple area) are within $-1.06\%$ and $1.23\%$; $28.14\%$ of Bitcoin returns are less than $-1.06\%$; and $31.54\%$ of Bitcoin returns are more than $1.23\%$. While S$\&$P 500 return distribution records about $80.05\%$ of S$\&$P 500 return within $-1.06\%$ and $1.23\%$; $10.76\%$ of S$\&$P 500 returns ( in cyan area) are less than $-1.06\%$; and $9.19\%$ of S$\&$P 500 returns ( in yellow area) are more than $1.23\%$.\\

\noindent
Table \ref{tab5} provides the empirical statistics from the sample data and the key statistics of GTS distribution. For S$\&$P 500 index, the sample size was $3386$ rate of returns with the largest and smallest values $8.96883\%$ and $-12.7652\%$ respectively. For Bitcoin, the sample size was $3705$ rate of returns with the largest and smallest values $28.0520\%$ and $-26.6197\%$, respectively. 
\begin{table}[ht]
\vspace{-0.5cm}
 \caption{Summary Statistics }
\label{tab5}
\centering
\begin{tabular}{@{}l | c c | c c@{}}
\toprule
                                              & \multicolumn{2}{c|}{\textbf{S\&P 500}} & \multicolumn{2}{c}{\textbf{Bitcoin}} \\ \toprule
\multirow{1}{*}{\textbf{Label}} & \multirow{1}{*}{\textbf{Empirical}} & \multirow{1}{*}{\textbf{Theoretical}}  & \multirow{1}{*}{\textbf{Empirical}} & \multirow{1}{*}{\textbf{Theoretical}} \\ \toprule
\multirow{1}{*}{Sample size(m)} & \multirow{1}{*}{3386} & \multirow{1}{*}{} & \multirow{1}{*}{3705} & \multirow{1}{*}{} \\ 
\multirow{1}{*}{Mean ($\kappa_{1}$)} & \multirow{1}{*}{0.0401\%} & \multirow{1}{*}{0.0401\%} & \multirow{1}{*}{0.1488\%} & \multirow{1}{*}{0.1489\%} \\ 
\multirow{1}{*}{Standard deviation ($\sigma$)} & \multirow{1}{*}{1.1192\%} & \multirow{1}{*}{1.0947\%} &\multirow{1}{*}{3.9784\%} & \multirow{1}{*}{3.9866\%} \\
\multirow{1}{*}{Coefficient of Variation (CV) } & \multirow{1}{*}{27.88917} & \multirow{1}{*}{27.27762} &\multirow{1}{*}{26.73285} & \multirow{1}{*}{26.77318} \\ 
\multirow{1}{*}{Skewness ($\frac{\kappa_{3}}{\kappa_{2}^{3/2}}$)} & \multirow{1}{*}{-0.72331} & \multirow{1}{*}{-0.57964} & \multirow{1}{*}{-0.32744} & \multirow{1}{*}{-0.31987} \\ 
\multirow{1}{*}{Kurtosis ($3+\frac{\kappa_{4}}{\kappa_{2}^{2}}$)} & \multirow{1}{*}{16.04641} & \multirow{1}{*}{8.92319} & \multirow{1}{*}{8.88504} & \multirow{1}{*}{9.74633} \\ 
\multirow{1}{*}{Max value} & \multirow{1}{*}{8.9683\%} & \multirow{1}{*}{} & \multirow{1}{*}{28.0520\%} & \multirow{1}{*}{} \\ 
\multirow{1}{*}{Min Value} & \multirow{1}{*}{-12.7652\%} & \multirow{1}{*}{} & \multirow{1}{*}{-26.6197\%} & \multirow{1}{*}{} \\
\bottomrule
\end{tabular}%
\vspace{-0.3cm}
\end{table}

\noindent
As shown in Table \ref{tab5}, the empirical and theoretical Mean ($\kappa_{1}$), Standard deviation ($\sigma$), and Coefficient of Variation (CV) are consistent for each asset. However, the empirical estimations overestimate the Kurtosis and skewness statistics, except for the S$\&$P 500 index; the plausible explanation is the impact of the outliers ($-12.7652\%$, $8.96883\%$).

\section{Analysis and Findings: Value-at-Risk and Average Value-at-Risk}
\noindent
The value-at-risk (VaR) and average value-at-risk (AVaR) are widely used financial risk measures. The value at risk (VaR) can be defined as the minimum level of loss or profit at a given confidence level. The estimation of the value-at-risk (VaR) and average value-at-risk (AVaR) are compared to the empirical VaR and AVaR. \\
\noindent
We consider the following sorted sample $x_{1}, x_{2}, \dots, x_{n}$ corresponding to instant $t_{1}, t_{2}, \dots, t_{n}$, the empirical VaR and AVaR at tail probability ($\alpha$) are obtained by applying the following estimators. \cite{kim2009computing, rachev2008advanced}. 
 \begin{equation}
\widehat{VaR}_{\alpha} =x_{(\lceil n\alpha \rceil)} \quad \quad \widehat{AVaR}_{\alpha} =\frac{1}{1-\alpha} \left[ \frac{1}{n}\sum_{\lceil n\alpha \rceil +1}^{n} x_{j} + ( \frac{\lceil n\alpha \rceil}{n} - \alpha ) x_{(\lceil n\alpha \rceil) } \right] 
\label{eq:l80}
\end{equation}
Where the notation $\lceil x \rceil$ stands for the smallest integer larger than $x$.\\
\subsection{Analysis and Findings: Value-at-Risk ($VaR_{\alpha}(X)$)}
\noindent
Based on the characteristic exponent of the GTS distribution, we develop a methodology to compute the theoretical value at risk (VaR). We assume the return distribution function is continuous; that is, there are no point masses. Formally, the VaR at confidence level $(1- \alpha)$ or tail probability ($\alpha$) is the $\alpha^{th}$ quantile of the return distribution and has the following mathematical expression.
 \begin{align}
VaR_{\alpha}(X)= inf\{ x | p(X \leq x) \geq \alpha \}= F^{-1}(x) \label{eq:l81}
\end{align} 

Where $0\leq \alpha \leq 1$ and $F^{-1}$ is the inverse of the cumulative distribution function. \\
\noindent
The GTS(\textbf{$\mu$}, \textbf{$\beta_{+}$}, \textbf{$\beta_{-}$}, \textbf{$\alpha_{+}$},\textbf{$\alpha_{-}$}, \textbf{$\lambda_{+}$}, \textbf{$\lambda_{-}$}) density and cumulative functions do not have closed form, which makes it difficult estimate of the $\alpha^{th}$ quantile analytically. Therefore, we rely on the computational method based on the Advanced FRFT developed previously. We consider a sequence of GTS cumulative function values $\left(F_{j}\right)_{1 \leq j \leq m}$  on a sequence of input value $\left(x_{j}\right)_{1 \leq j \leq m}$. $\left(F_{j}\right)_{1 \leq j \leq m}$ is generated by the Advanced FRFT scheme. We also consider $x_{\alpha}$ the $\alpha^{th}$ quantile defined by $F(x_{\alpha})=\alpha$ with $x_{i}<x_{\alpha}<x_{i+1}$ and $F_{i}<F(x_{\alpha})<F_{i+1}$.\\
By applying the Taylor series approximation, we have :
  \begin{equation} 
 \begin{aligned}
F(x_{\alpha})&=F_{i} + \frac{u}{1!}F^{(1)}_{i} + \frac{u^{2}}{2!}F^{(2)}_{i} + \frac{u^{3}}{3!}F^{(3)}_{i} + \frac{u^{4}}{4!}F^{(4)}_{i} + O(u^{4})\\
 u&=x_{\alpha} - x_{i} \quad \quad \lim_{u\to 0} O(u^{4})=0
\label{eq:l82}
\end{aligned}
\end{equation}

\noindent
By applying the central difference representations of $O(\Delta^{2})$ with $\Delta =x_{i+1} - x_{i}$, we have the values of first, second, third, and fourth-order derivatives.
  \begin{equation} 
 \begin{aligned}
a_{1}&=\Delta  F^{(1)}_{i}=\frac{-F_{i-1} + F_{i+1}}{2} \quad \quad a_{3}=\Delta^{3} F^{(3)}_{i}=\frac{-F_{i-2} + 2F_{i-1} - 2F_{i+1} + F_{i+2}}{2}\\
a_{2}&=\Delta^{2} F^{(2)}_{i}=F_{i-1} - 2F_{i} + F_{i+1} \quad \quad  a_{4}=\Delta^{4} F^{(4)}_{i}=F_{i-2} - 4F_{i-1} + 6F_{i} - 4F_{i+1} + F_{i+2}
\label{eq:l83}
\end{aligned}
\end{equation}
\noindent
By removing the function $O(x_{\alpha} - x_{i})$ in (\ref{eq:l82}), we have the following polynomial equation:
  \begin{equation} 
 \begin{aligned}
 -(\alpha -F_{i}) + a_{1}\frac{x_{\alpha} - x_{i}}{x_{i+1} - x_{i}}  + \frac{a_{2}}{2!}(\frac{x_{\alpha} - x_{i}}{x_{i+1} - x_{i}})^{2}+ \frac{a_{3} }{3!}(\frac{x_{\alpha} - x_{i}}{x_{i+1} - x_{i}})^{3}+ \frac{a_{4}}{4!}(\frac{x_{\alpha} - x_{i}}{x_{i+1} - x_{i}})^{4} = 0 
\label{eq:l84}
\end{aligned}
\end{equation}

\noindent
Let us take $b_{0}=-(\alpha -F_{i})$, $b_{1}=a_{1}$, $b_{2}=\frac{a_{2}}{2!}$, $b_{3}=\frac{a_{3}}{3!}$, $b_{4}=\frac{a_{4}}{4!}$ and $y=\frac{x_{\alpha} - x_{i}}{x_{i+1} - x_{i}}$. The polynomial equation becomes:
 \begin{equation} 
b_{0} + b_{1} y  + b_{2} y^{2} + b_{3} y^{3} + b_{4} y^{4} = 0 \\
\label{eq:l85}
\end{equation}

\noindent 
The Intermediate Value Theorem guarantees the existence of the solution $(y)$ over the interval (0,1). For a root (y) of the equation (\ref{eq:l85}) with $0<y<1$, we have the estimation of the $\alpha^{th}$ quantile as follows.
 \begin{equation} 
 x_{\alpha}=x_{i} + y(x_{i+1} - x_{i})
\label{eq:l86}\\
\end{equation}

\noindent 
It was shown in Fig\ref{fig:13} that the GTS distribution generated from Bitcoin and S$\&$P500 returns are not symmetric. Therefore, for a given tail probability ($\alpha$), the $VaR_{\alpha}(X)$ is not expected to yield the same value as the $VaR_{1-\alpha}(X)$ for a confident level ($1-\alpha$). For tail probability ($\alpha$) from $0.5\%$, $1\%$,\dots, $10\%$, the theoretical and empirical Value-at-Risk ($VaR_{\alpha}(X)$) was computed and summarized in table \ref{tab8} (appendix \ref{eq:an2}). For the correspondent confidence level ($1-\alpha$) from $90\%$,\dots ,$99\%$, $99.5\%$ the estimations of the the theoretical and empirical Value-at-Risk ($VaR_{1-\alpha}(X)$) are summarised in table \ref{tab6}. Both tables show that the empirical and theoretical Value-at-Risk is consistent. As expected, the bitcoin theoretical and empirical Value-at-Risk ($VaR_{1-\alpha}(X)$) are higher than that of the S$\&$P 500 index, which is consistent with the heavy-tailedness of the bitcoin return. We have the same pattern in table \ref{tab8}. 
\begin{table}[ht]
\vspace{-0.4cm}
 \caption{ Value-at-Risk Statistics }
\label{tab6}
\centering
\begin{tabular}{@{}c | c c | c c@{}}
\toprule
  \multicolumn{1}{c|}{\textbf{$VaR_{1-\alpha}(X)$}} & \multicolumn{2}{c|}{\textbf{S\&P 500 index (\%)}} & \multicolumn{2}{c}{\textbf{Bitcoin (\%)}} \\ \toprule
\multirow{1}{*}{\textbf{Confidence Level ($1-\alpha$)}} & \multirow{1}{*}{\textbf{Empirical}} & \multirow{1}{*}{\textbf{Theoretical}}  & \multirow{1}{*}{\textbf{Empirical}} & \multirow{1}{*}{\textbf{Theoretical}} \\ \toprule
\multirow{1}{*}{$90\%$} & \multirow{1}{*}{1.1819} & \multirow{1}{*}{1.1760} &\multirow{1}{*}{4.3173} & \multirow{1}{*}{4.2392} \\ 
\multirow{1}{*}{$91\%$} & \multirow{1}{*}{1.2638} & \multirow{1}{*}{1.2499} & \multirow{1}{*}{4.5631} & \multirow{1}{*}{4.5379} \\
\multirow{1}{*}{$92\%$} & \multirow{1}{*}{1.3311} & \multirow{1}{*}{1.3333} & \multirow{1}{*}{4.8620} & \multirow{1}{*}{4.8761} \\
\multirow{1}{*}{$93\%$} & \multirow{1}{*}{1.4050} & \multirow{1}{*}{1.4288} & \multirow{1}{*}{5.1729} & \multirow{1}{*}{5.2647} \\ 
\multirow{1}{*}{$94\%$} & \multirow{1}{*}{1.4770} & \multirow{1}{*}{1.5402} & \multirow{1}{*}{5.5775} & \multirow{1}{*}{5.7202} \\ 
\multirow{1}{*}{$95\%$} & \multirow{1}{*}{1.5939} & \multirow{1}{*}{1.6738} &\multirow{1}{*}{6.3628} & \multirow{1}{*}{6.2679} \\ 
\multirow{1}{*}{$96\%$} & \multirow{1}{*}{1.7231} & \multirow{1}{*}{1.8399} & \multirow{1}{*}{7.1021} & \multirow{1}{*}{6.9508} \\ 
\multirow{1}{*}{$97\%$} & \multirow{1}{*}{1.9259} & \multirow{1}{*}{2.0583} &\multirow{1}{*}{7.9802} & \multirow{1}{*}{7.8507} \\ 
\multirow{1}{*}{$98\%$} & \multirow{1}{*}{2.2504} & \multirow{1}{*}{2.3738} & \multirow{1}{*}{9.5034} & \multirow{1}{*}{9.1534} \\
\multirow{1}{*}{$99\%$} & \multirow{1}{*}{2.8115} & \multirow{1}{*}{2.9334} & \multirow{1}{*}{11.1161} & \multirow{1}{*}{11.4653} \\
\multirow{1}{*}{$99.5\%$} & \multirow{1}{*}{3.3713} & \multirow{1}{*}{3.5166} & \multirow{1}{*}{13.5660} & \multirow{1}{*}{13.8771} \\ 
\bottomrule
\end{tabular}%
\vspace{-0.3cm}
\end{table}

\noindent
As shown in table \ref{tab6}, The 95\% $VaR_{1-\alpha}(X)$ of both S$\&$P 500 index and Bitcoin are equal to $3.52\%$ and $13.88\%$ respectively. That is, Bitcoin gains more than 13.88\% of its present value with a probability of $5\%$, whereas the S$\&$P 500 index gains only more than $3.52\%$ of its present value with a probability of $5\%$. As illustrated in Fig \ref{fig:15} and Fig \ref{fig:16}, the value-at-risk (\textbf{$VaR_{\alpha}(X)$}) of bitcoin and S$\&$P 500 index increase at an increasing rate in each interval. There is also a discrepancy between the value-at-risk (\textbf{$VaR_{\alpha}(X)$}) of bitcoin and S$\&$P 500 index. A disadvantage of value-at-risk (VaR)\cite{rachev2008advanced} is that it does not provide any information about the magnitude of the losses or profits larger than the value-at-risk level.

\subsection{Analysis and Findings: Average Value-at-Risk ($AVaR_{\alpha}(X)$)}
\noindent
The average value-at-risk ($AVaR_{\alpha}$) at tail probability $\alpha$ is defined as the average of the value-at-risks (VaRs) that are larger than the value-at-risk (VaR) at tail probability $\alpha$. The mathematical expression can be written as follows\\
 \begin{equation}
   AVaR_{\alpha}(X)=\frac{1}{\alpha}\int_{0}^{\alpha}VaR_{y}(X)dy
   \label {eq:l87}
\end{equation}

\noindent
The Generalized Tempered Stable (GTS) distribution is a continuous variable ($X$), and the Value-at-Risk (VaR) has a simplified expression $VaR_{\alpha}(X)=F_{X}^{-1}(\alpha)$ in (\ref{eq:l81}) and the integral in (\ref{eq:l87}) becomes
 \begin{align}
  \int_{0}^{\alpha}VaR_{y}(X) dy= \int_{0}^{\alpha}F_{X}^{-1}(y)dy=\int_{0}^{F_{X}^{-1}(\alpha)}\theta dF_{X}(\theta)=E\left[X \mathbbm{1}_{\{X\leq VaR_{\alpha}(X)\}}\right] \label {eq:l88}
\end{align}
\noindent
The average value-at-risk (\ref{eq:l87}) becomes
 \begin{equation}
 \begin{aligned}
   AVaR_{\alpha}(X)&=\frac{1}{\alpha}E\left[X \mathbbm{1}_{\{X\leq VaR_{\alpha}(X)\}}\right]=\frac{1}{\alpha}E\left[VaR_{\alpha}(X) \mathbbm{1}_{\{X\leq VaR_{\alpha}(X)\}}  + \left(X - VaR_{\alpha}(X)\right)^{-} \right]\\ 
   &=VaR_{\alpha}(X) + \frac{1}{\alpha}E\left[ \left(X - VaR_{\alpha}(X)\right)^{-} \right]
   \label {eq:l89}
\end{aligned}
\end{equation}

\noindent
For the confident level $(1-\alpha)$, the average value-at-risk ($AVaR_{1-\alpha}$) becomes
 \begin{equation}
 \begin{aligned}
   AVaR_{1-\alpha}(X)&=\frac{1}{1-\alpha}\int_{0}^{1-\alpha}VaR_{1-y}(X)dy=\frac{1}{1-\alpha}\int_{\alpha}^{1}VaR_{y}(X)dy\\ 
   &=VaR_{1-\alpha}(X) + \frac{1}{1-\alpha}E\left[ \left(X - VaR_{1-\alpha}(X)\right)^{+} \right]
   \label {eq:l90}
\end{aligned}
\end{equation}
\noindent
We have the following expression on the loss ( $\alpha \leq \frac{1}{2}$) and the profit ( $\alpha \geq \frac{1}{2}$) of the return distribution
 \begin{equation}
 \begin{aligned}
 AVaR_{\alpha}(X) &=VaR_{\alpha}(X) + \frac{1}{\alpha}E\left[ \left(X - VaR_{\alpha}(X)\right)^{-} \right] \hspace{5mm}  &\hbox{$\alpha < \frac{1}{2}$}\\
   AVaR_{1-\alpha}(X)&=VaR_{1-\alpha}(X) + \frac{1}{1-\alpha}E\left[ \left(X - VaR_{1-\alpha}(X)\right)^{+} \right] \hspace{5mm}  &\hbox{$\alpha > \frac{1}{2}$}
   \label {eq:l91}
\end{aligned}
\end{equation}
 \begin{theorem} \label{lem1}\ \\
 \noindent
 Let $ k \in \mathbb{R} $. $X$ is a GTS(\textbf{$\mu$}, \textbf{$\beta_{+}$}, \textbf{$\beta_{-}$}, \textbf{$\alpha_{+}$},\textbf{$\alpha_{-}$}, \textbf{$\lambda_{+}$}, \textbf{$\lambda_{-}$}) random variable with characteristic exponent function $\Psi(\xi)$ .\\
\noindent
There exists $q > 0$ such that 
 \begin{equation}
 \begin{aligned}
E[\left(X - k\right)^{+}]&=\frac{1}{2\pi}\int_{-\infty + iq}^{+\infty + iq}\frac{-1}{z^2}e^{izk + \Psi(-z)} dz  \\
E[\left(X - k\right)^{-}]&=\frac{1}{2\pi}\int_{-\infty - iq}^{+\infty - iq}\frac{1}{z^2}e^{izk + \Psi(-z)} dz  \label{eq:l92}
\end{aligned}
\end{equation}

\end{theorem}
\begin{proof}[Proof:] \ \\
\noindent
$g(X,k)=(X - k)^{+}$ is the payoff of the call option, and the Fourier transform of $g$ is derived as follows:
 \begin{equation}
 \begin{aligned}
\scrF[g] (y,k)&=\int_{-\infty}^{+\infty}e^{-iyx}g(x,k)dx =\int_{-\infty}^{+\infty}e^{-iyx}{(x - k)^{+}}dx=\int_{k}^{+\infty}e^{-iyx}{(x - k)}dx\\
&=\left[\frac{1 + (x-k)iy}{y^2}e^{-iyx}\right]_{k}^{+\infty}=-\frac{1}{y^2}e^{-iyk} \hspace{5mm}  \hbox{for $\Im(y) < 0$} 
 \label {eq:l93}
\end{aligned}
\end{equation}
Where $\Im(y)$ is the imaginary of $y$.\\
If the  payoff function becomes the payoff of the put option, we have  $g(X,k)=(X - k)^{-}$  and the Fourier transform becomes
 \begin{equation}
  \scrF[g] (y,k)=\left[\frac{1 + (x-k)iy}{y^2}e^{-iyx}\right]_{- \infty}^{k}=\frac{1}{y^2}e^{-iyk}  \hspace{5mm}  \hbox{for $\Im(y) > 0$} 
   \label {eq:l94}
\end{equation}
\noindent
The call payoff (\ref{eq:l92}) can be recovered from the inverse of Fourier if there exists $q>0$ such that 
  \begin{align}
    \check{g}(x,k)=\frac{1}{2\pi}\int_{-\infty + iq}^{+\infty + iq}e^{iyx}\scrF[g] (y,k)dy=-\frac{1}{2\pi}\int_{-\infty + iq}^{+\infty + iq}\frac{1}{y^2}e^{iy(x-k)}dy  
 \label {eq:l95}
\end{align}

  \begin{align*}
E[\left(X - k\right)^{+}]&=\int_{-\infty}^{+\infty} g(y,k)f(y)dy=\frac{1}{2\pi}\int_{-\infty}^{+\infty}\int_{-\infty}^{+\infty}e^{-i y z + \Psi(z)} g(y,k)dydz  \quad   \quad   g(y,k)=(y - k)^{+}\\
&=\frac{1}{2\pi}\int_{-\infty}^{+\infty}e^{\Psi(z)}\int_{-\infty}^{+\infty}e^{-i y z } g(y,k)dydz= \frac{1}{2\pi}\int_{-\infty - iq}^{+\infty - iq}e^{\Psi(z)} \scrF[g] (z,k)dz   \quad  \hbox{recall (\ref{eq:l93})}\\ 
&=-\frac{1}{2\pi}\int_{-\infty - iq}^{+\infty - iq}\frac{1}{z^2}e^{-izk + \Psi(z)} dz \\
&=\frac{1}{2\pi}\int_{-\infty + iq}^{+\infty + iq}\frac{-1}{z^2}e^{izk + \Psi(-z)} dz   
\end{align*}
The same development holds for $E[\left(X - k\right)^{-}]$.
  \end{proof} 
 \noindent
 The error function $ER(k,q)$ between the call payoff and the inverse Fourier (\ref{eq:l95}) is defined as follows. 
 
   \begin{align}
ER(k,q)=\sqrt{ \frac{1}{m}\sum_{j=1}^{m} \left[ ( x_{j} - k)^{+} -  \check{g}(x_{j},k) \right]^2}  \hspace{5mm}  \hbox{with $-M \leq x_{j}\leq M$ \ ,  \ $M>0$} \label{eq:l96}
\end{align}
\noindent
Where $k$ (strike price) and $q$ (parameter) are the inputs of the function $ER(k,q)$. The parameter $q$ is used to optimize the function $ER(k,q)$. Fig \ref{fig15a} shows how the accuracy of the inverse Fourier (\ref{eq:l95}) depends on the parameter $q$.
\begin{figure}[ht]
\vspace{-0.3cm}
    \centering
  \begin{subfigure}[b]{0.42\linewidth}
    \includegraphics[width=\linewidth]{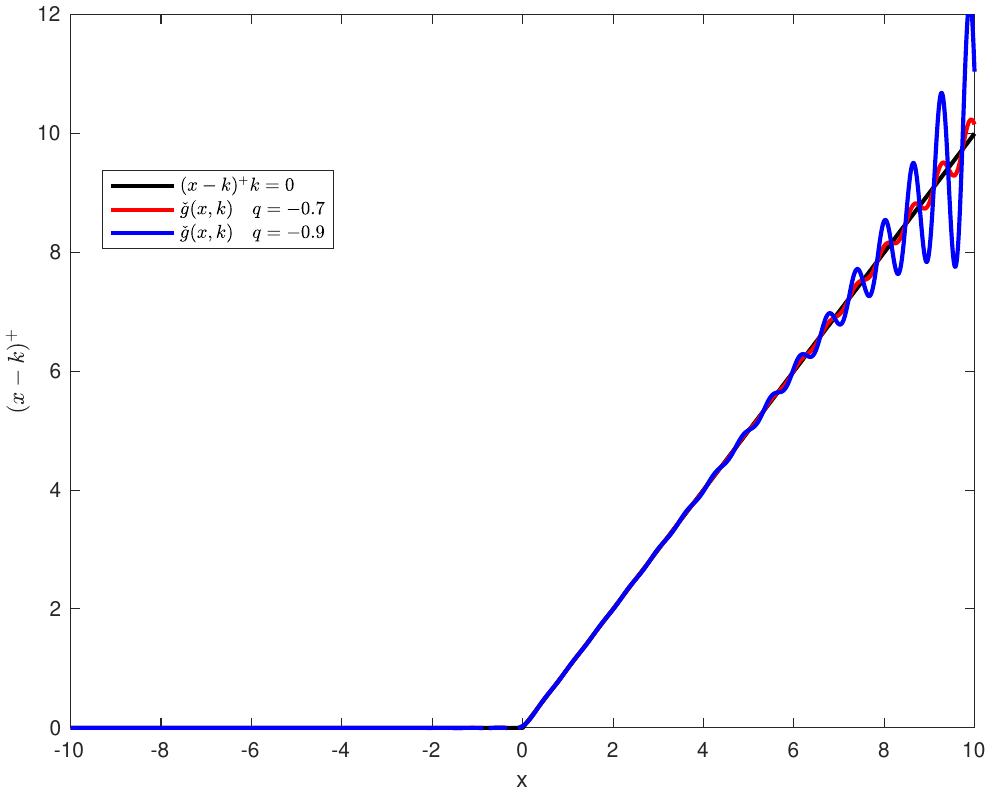}
\vspace{-0.5cm}
     \caption{$(X - k)^{+}$ versus $\check{g}(x,k)$}
         \label{fig15a}
  \end{subfigure}
\hspace{-0.2cm}
  \begin{subfigure}[b]{0.48 \linewidth}
    \includegraphics[width=\linewidth]{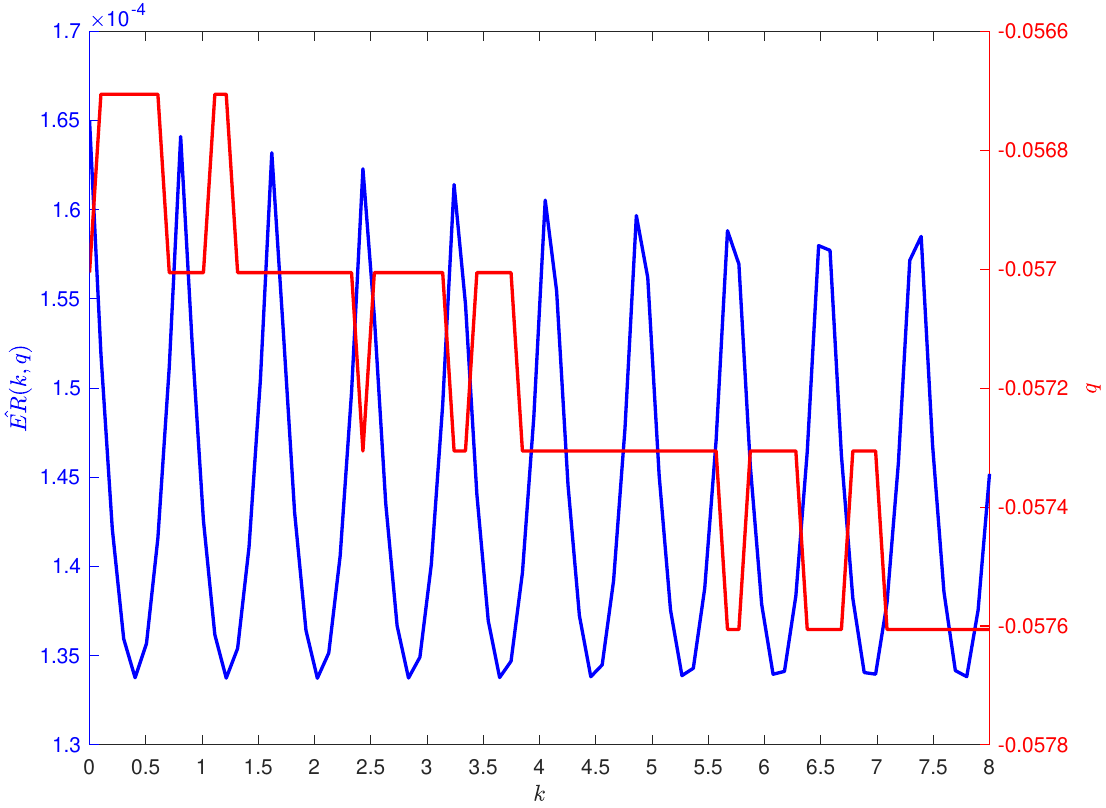}
\vspace{-0.5cm}
     \caption{$ER(k,q)$ and Optimal parameter (q)}
         \label{fig15b}
          \end{subfigure}
\vspace{-0.6cm}
  \caption{Optimal parameter ($q$) and minimum Error value ($ER(k,q)$)}
  \label{fig15ab}
\vspace{-0.3cm}
\end{figure}

\noindent
Fig \ref{fig15b} displays the $ER(k,q)$ minimum value (in blue color) as a function of the strike price ($k$); the correspondent optimal parameter ($q$) is graphed as a function of the strike price ($k$) in Fig \ref{fig15b}. The $ER(k,q)$ minimum value oscillates between $1.337*10^{-4}$ and $1.650*10^{-4}$, which is almost zero. The optimal parameter ($q$) decreases slowly from $-0.0576$ to $-0.0567$ as shown Fig \ref{fig15b}.\\
\newpage
 
  \begin{corollary}\label{lem2}\ \\
  \noindent
$X$ is a GTS(\textbf{$\mu$}, \textbf{$\beta_{+}$}, \textbf{$\beta_{-}$}, \textbf{$\alpha_{+}$},\textbf{$\alpha_{-}$}, \textbf{$\lambda_{+}$}, \textbf{$\lambda_{-}$}) random variable with characteristic exponent function $\Psi(\xi)$.\\
\noindent
There exists $q>0$ such that 
  \begin{equation}
 \begin{aligned}
AVaR_{1-\alpha}(X) =VaR_{1-\alpha}(X) + \frac{1}{1-\alpha}\frac{1}{2\pi}\int_{-\infty + iq}^{+\infty + iq}\frac{-1}{z^2}e^{izVaR_{1-\alpha}(X) + \Psi(-z)} dz  \hspace{5mm}  &\hbox{$\alpha > \frac{1}{2}$} \\
AVaR_{\alpha}(X) =VaR_{\alpha}(X) + \frac{1}{\alpha}\frac{1}{2\pi}\int_{-\infty - iq}^{+\infty - iq}\frac{1}{z^2}e^{izVaR_{\alpha}(X) + \Psi(-z)} dz  \hspace{20mm}  &\hbox{$\alpha < \frac{1}{2}$} \label{eq:l97}
\end{aligned}
\end{equation}
\end{corollary}
\begin{proof}[Proof:] \ \\
Equations ($\ref{eq:l91}$) lead to equation ($\ref{eq:l97}$) by substituting $k=VaR_{\alpha}(X)$ and applying theorem ($\ref{lem1}$).
\end{proof} 

\noindent 
 For tail probability ($\alpha$) from $0.5\%$, $1\%$,\dots, $10\%$, the theoretical and empirical Value-at-Risk ($VaR_{\alpha}(X)$) was computed and summarized in table \ref{tab9} (appendix \ref{eq:an2}). For the correspondent confidence Level ($1-\alpha$) from $90\%$,\dots, $99\%$, $99.5\%$ the estimations of the the theoretical and empirical Value-at-Risk ($VaR_{1-\alpha}(X)$) are summarised in table \ref{tab7}. Both tables show that the empirical and theoretical Value-at-Risk are consistent. As expected, the bitcoin theoretical and empirical Value-at-Risk ($VaR_{1-\alpha}(X)$) are higher than that of the S$\&$P 500 index, which is consistent with the heavy-tailedness of the bitcoin.

\begin{table}[ht]
\vspace{-0.5cm}
 \caption{ Average Value-at-Risk Statistics }
\label{tab7}
\centering
\begin{tabular}{@{}c | c c | c c@{}}
\toprule
  \multicolumn{1}{c|}{\textbf{$AVaR_{1-\alpha}(X)$}} & \multicolumn{2}{c|}{\textbf{S\&P 500 (\%)}} & \multicolumn{2}{c}{\textbf{Bitcoin (\%)}} \\ \toprule
\multirow{1}{*}{\textbf{Confidence Level ($\alpha$)}} & \multirow{1}{*}{\textbf{Empirical}} & \multirow{1}{*}{\textbf{Theoretical}}  & \multirow{1}{*}{\textbf{Empirical}} & \multirow{1}{*}{\textbf{Theoretical}} \\ \toprule
\multirow{1}{*}{$90\%$} & \multirow{1}{*}{1.8976} & \multirow{1}{*}{1.9278} &\multirow{1}{*}{7.3077} & \multirow{1}{*}{7.3190} \\ 
\multirow{1}{*}{$91\%$} & \multirow{1}{*}{1.9730} & \multirow{1}{*}{2.0074} & \multirow{1}{*}{7.6273} & \multirow{1}{*}{7.6451} \\
\multirow{1}{*}{$92\%$} & \multirow{1}{*}{2.0584} & \multirow{1}{*}{2.0971} & \multirow{1}{*}{7.9916} & \multirow{1}{*}{8.0130} \\
\multirow{1}{*}{$93\%$} & \multirow{1}{*}{2.1585} & \multirow{1}{*}{2.1997} & \multirow{1}{*}{8.4150} & \multirow{1}{*}{8.4343} \\ 
\multirow{1}{*}{$94\%$} & \multirow{1}{*}{2.2790} & \multirow{1}{*}{2.3193} & \multirow{1}{*}{8.9274} & \multirow{1}{*}{8.9259} \\ 
\multirow{1}{*}{$95\%$} & \multirow{1}{*}{2.4304} & \multirow{1}{*}{2.4624} &\multirow{1}{*}{9.5154} & \multirow{1}{*}{9.5145} \\ 
\multirow{1}{*}{$96\%$} & \multirow{1}{*}{2.6269} & \multirow{1}{*}{2.6399} & \multirow{1}{*}{10.2239} & \multirow{1}{*}{10.2448} \\ 
\multirow{1}{*}{$97\%$} & \multirow{1}{*}{2.8887} & \multirow{1}{*}{2.8724} &\multirow{1}{*}{11.1451} & \multirow{1}{*}{11.2015} \\ 
\multirow{1}{*}{$98\%$} & \multirow{1}{*}{3.3001} & \multirow{1}{*}{3.2070} & \multirow{1}{*}{12.3666} & \multirow{1}{*}{12.5767} \\
\multirow{1}{*}{$99\%$} & \multirow{1}{*}{4.1403} & \multirow{1}{*}{3.7960} & \multirow{1}{*}{14.5488} & \multirow{1}{*}{14.9924} \\
\multirow{1}{*}{$99.5\%$} & \multirow{1}{*}{5.2054} & \multirow{1}{*}{4.4047} & \multirow{1}{*}{16.9509} & \multirow{1}{*}{17.4790} \\ \bottomrule
\end{tabular}%
\end{table}
\newpage
\noindent
To generalize the computation performed in table \ref{tab7} \& \ref{tab9} and accounted for a large range of values, we consider the interval $[0,10]$ for tail probability ($\alpha$) in Fig \ref{fig:16}; and the interval $[90,100]$ for confidence level ($1-\alpha$) in Fig \ref{fig:15}. As illustrated in Fig \ref{fig:15}, the average value-at-risk (\textbf{$AVaR_{1-\alpha}(X)$}) of bitcoin and S$\&$P 500 index increase at an increasing rate on both intervals, which justified concave nature of each curve. 
  \begin{figure}[ht]
     \centering
     \begin{subfigure}[b]{0.45\textwidth}
         \centering
         \includegraphics[width=\textwidth]{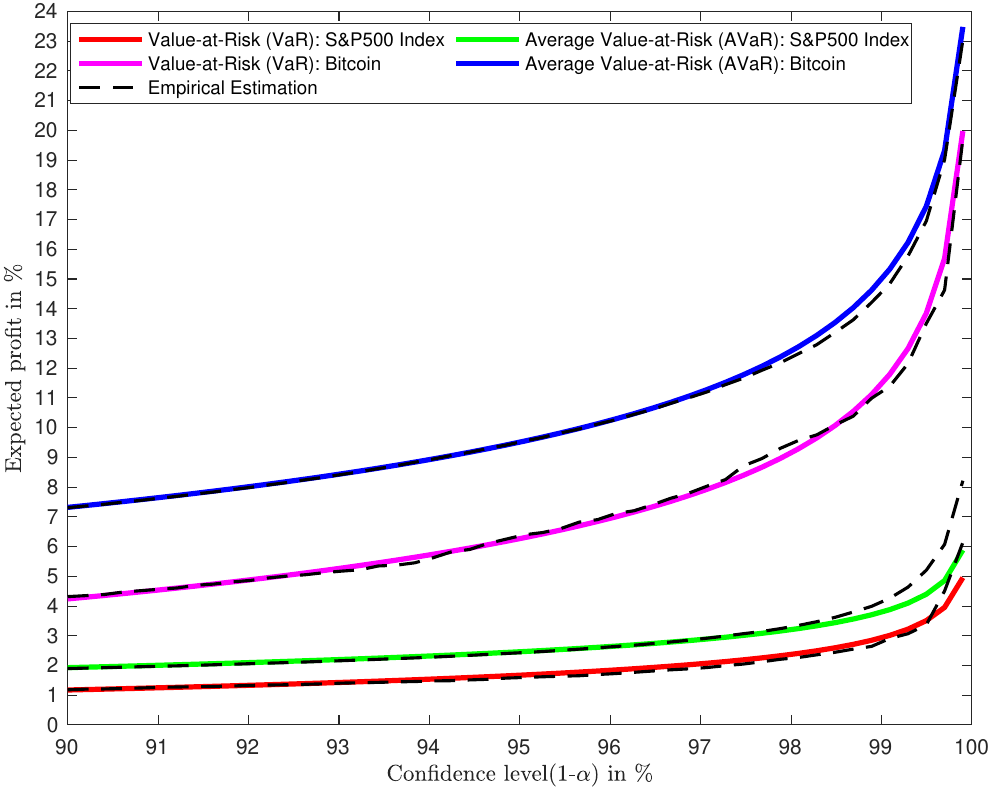}
         \vspace{-0.6cm}
         \caption{\textbf{$VaR_{1-\alpha}(X)$} versus  \textbf{$AVaR_{1-\alpha}(X)$}}
         \label{fig:15}
     \end{subfigure}
     \begin{subfigure}[b]{0.45\textwidth}
         \centering
         \includegraphics[width=\textwidth]{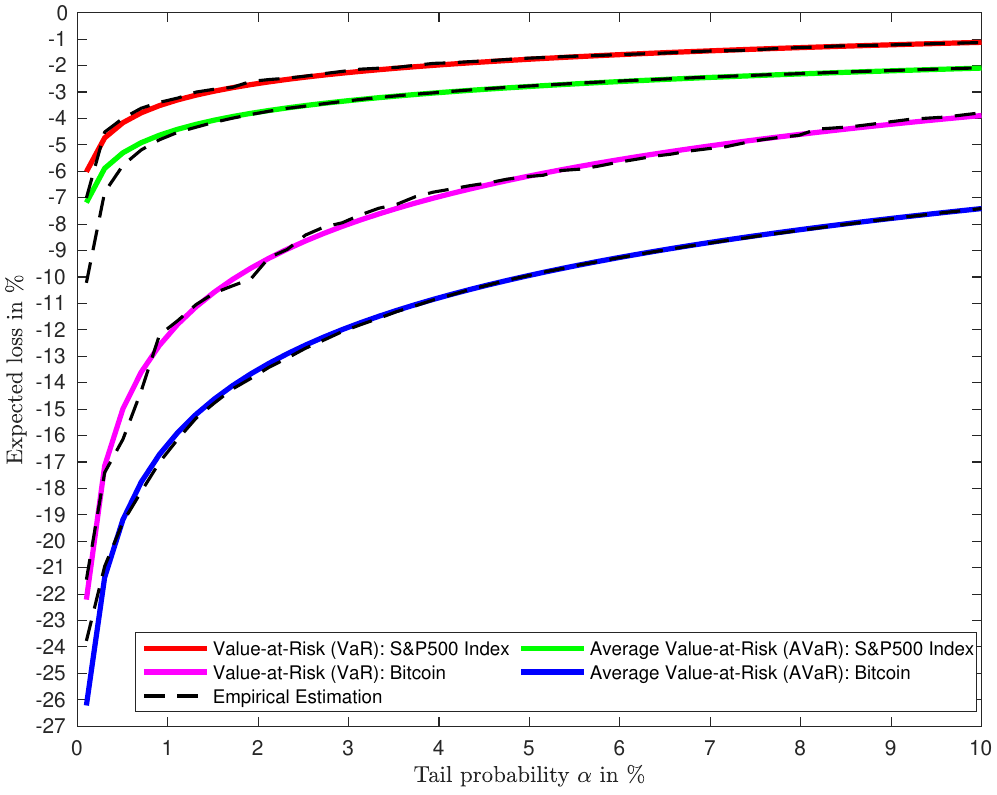}
         \vspace{-0.6cm}
         \caption{\textbf{$VaR_{\alpha}(X)$} versus  \textbf{$AVaR_{\alpha}(X)$}}
         \label{fig:16}
     \end{subfigure}
        \vspace{-0.3 cm}
        \caption{Average Value-at-Risk Statistics}
        \label{fig:17}
       \vspace{-0.3cm}
\end{figure}

\noindent
The \textbf{$AVaR_{1-0.999}(X)=23.48\%$} for bitcoin and \textbf{$AVaR_{1-0.999}(X)=5.88\%$} for the S$\&$P 500 index; and the tail probability($\alpha$), the \textbf{$AVaR_{0.001}(X)=-26.21\%$} for bitcoin and \textbf{$AVaR_{0.001}(X)=-7.18\%$} for the S$\&$P 500 index. As shown in Fig \ref{fig:17}, the left side of the tail probability generates higher value-at-risk (\textbf{$VaR_{\alpha}(X)$}) and average value-at-risk (\textbf{$AVaR_{\alpha}(X)$}) than the right side of tail probability for both bitcoin and S$\&$P 500 index. At a risk level ($\alpha$), the severity of the loss (\textbf{$AVaR_{\alpha}(X)$}) on the left side of the distribution is larger than the severity of the profit (\textbf{$AVaR_{1-\alpha}(X)$}) on the right side of the distribution. The result comes from the left-skewness nature of the distribution for both bitcoin and S$\&$P 500 index.\\

\noindent
The magnitude of the discrepancy between the average value-at-risk (\textbf{$AVaR_{\alpha}(X)$}) of bitcoin and S$\&$P 500 index can be evaluated by dividing the \textbf{$AVaR_{\alpha}(X)$} of the Bitcoin by the \textbf{$AVaR_{\alpha}(X)$} of the S$\&$P 500 index and analyzed the relationship as shown in Fig \ref{fig:18}.

 \begin{figure}[ht]
    \centering
         \includegraphics[width=0.45\textwidth]{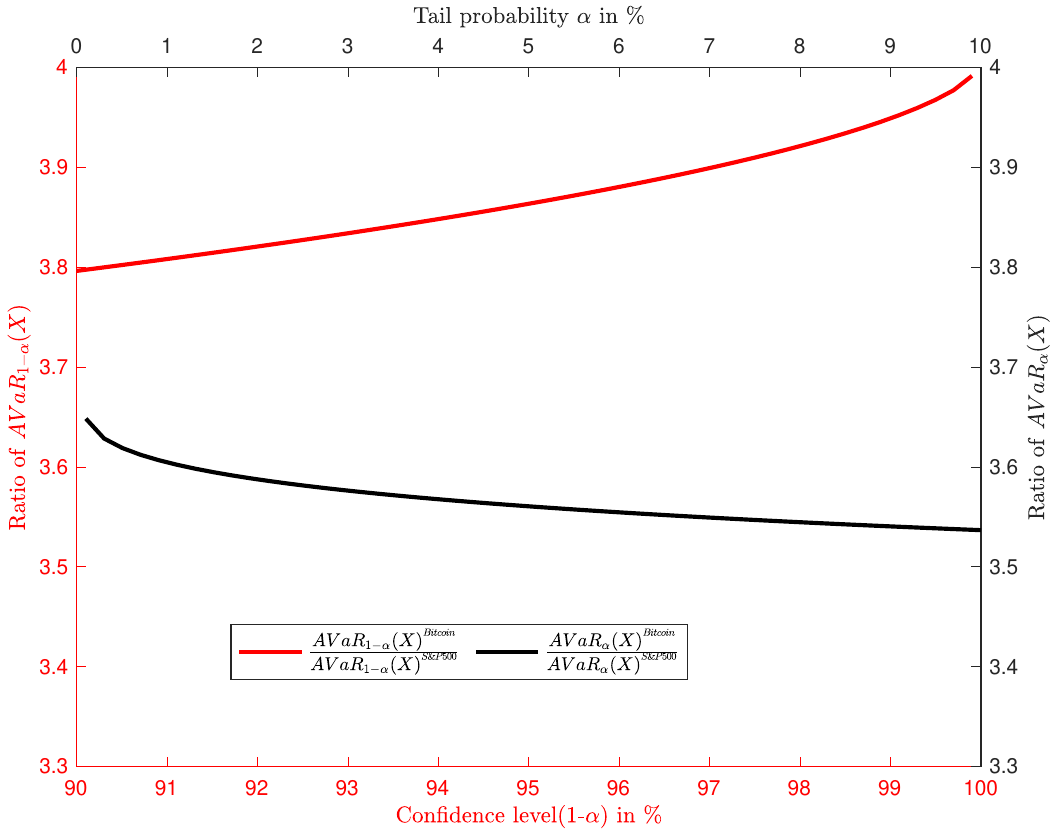}
         \vspace{-0.2cm}
         \caption{$\frac{\textbf{$AVaR_{\alpha}(X)$}^{Bitcoin}}{\textbf{$AVaR_{\alpha}(X)$}^{S\&P 500}}$}
         \label{fig:18}
         \vspace{-0.3cm}
\end{figure}

\noindent
The Profits generated by the average value-at-risk (\textbf{$AVaR_{1-\alpha}(X)$}) of bitcoin at one significant figure is four times larger than that of the S$\&$P 500 index. Similarly, the losses generated by the average value-at-risk (\textbf{$AVaR_{\alpha}(X)$}) of bitcoin at one significant figure is four times larger than that of the S$\&$P 500 index as shown in Fig \ref{fig:18}. 


\section {Conclusion} 
\noindent
The paper analyzes the daily return distribution of the bitcoin and S$\&$P 500 index. It assesses their tail probabilities through the value-at-risk (VaR) and the average value-at-risk (AVaR). As a methodology, We use the historical prices for Bitcoin and S$\&$P 500 index. The Bitcoin price spans from January 04, 2010, to June 16, 2023, and the S$\&$P 500 index price period is from April 28, 2013, to June 22, 2023. Each historical data is used to fit the seven-parameter General Tempered Stable (GTS) distribution to the underlying data return distribution. The advanced Fast Fractional Fourier transform (FRFT) scheme is developed by combining the classic Fast Fractional Fourier (FRFT) algorithm with the 12-point rule Composite Newton-Cotes Quadrature. The advanced FRFT scheme is used to perform the Maximum likelihood estimation of the seven parameters of the GTS distribution and the computations derived from the GTS density and cumulation distribution functions. It results from the GTS distribution fitting that the stability indexes, the process intensities, and the decay rate are all positive, and the bitcoin return and S$\&$P 500 index return are infinite activity processes with infinite jumps in any given time interval. The parameter analysis shows that the Bitcoin and S$\&$P 500 index returns are left-skewed distributions. The study of the probability density function (pdf) and some Key Statistics show that the tail events of both Bitcoin and S$\&$P 500 index are much more prevalent than we would expect with a Normal distribution. Both probability density functions are leptokurtic distributions. However, The heavy-tailedness is the main characteristic of the Bitcoin return, whereas the peakedness is the main characteristic of the S$\&$P 500 index return. The GTS distribution shows that $80.05\%$ of S$\&$P 500 returns are within $-1.06\%$ and $1.23\%$ against only $40.32\%$ of Bitcoin returns. The value-at-risk and the average value-at-risk reveal significant differences in tail probability between the Bitcoin returns and S$\&$P 500 index returns. At a risk level ($\alpha$), the severity of the loss (\textbf{$AVaR_{\alpha}(X)$}) on the left side of the distribution is larger than the severity of the profit (\textbf{$AVaR_{1-\alpha}(X)$}) on the right side of the distribution.  Compared to the S$\&$P 500 index, Bitcoin has $39.73\%$ more prevalence to produce high daily returns (more than $1.23\%$ or less than $-1.06\%$). The severity analysis shows that at a risk level ($\alpha$) the average value-at-risk (\textbf{$AVaR(X)$}) of the bitcoin returns at one significant figure is four times larger than that of the S$\&$P 500 index returns at the same risk.

\bibliographystyle{unsrt}
\nocite{*}
\bibliography{nzokem_pricing}

\clearpage
\newpage\appendix
\renewcommand{\thesection}{\Alph{section}.\arabic{section}}
\setcounter{section}{0}
\section{GTS(\textbf{$\beta_{+}$},\textbf{$\beta_{-}$},\textbf{$\alpha_{+}$},\textbf{$\alpha_{-}$},\textbf{$\lambda_{+}$},\textbf{$\lambda_{-}$}) Parameter Estimation \\
by Newton – Raphson Iteration Algorithm (\ref{eq:l38})}\label{eq:an1}
\begin{figure}[ht]
\vspace{-0.6cm}
    \centering
\hspace{-0.3cm}
  \begin{subfigure}[b]{0.45\linewidth}
    \includegraphics[width=\linewidth]{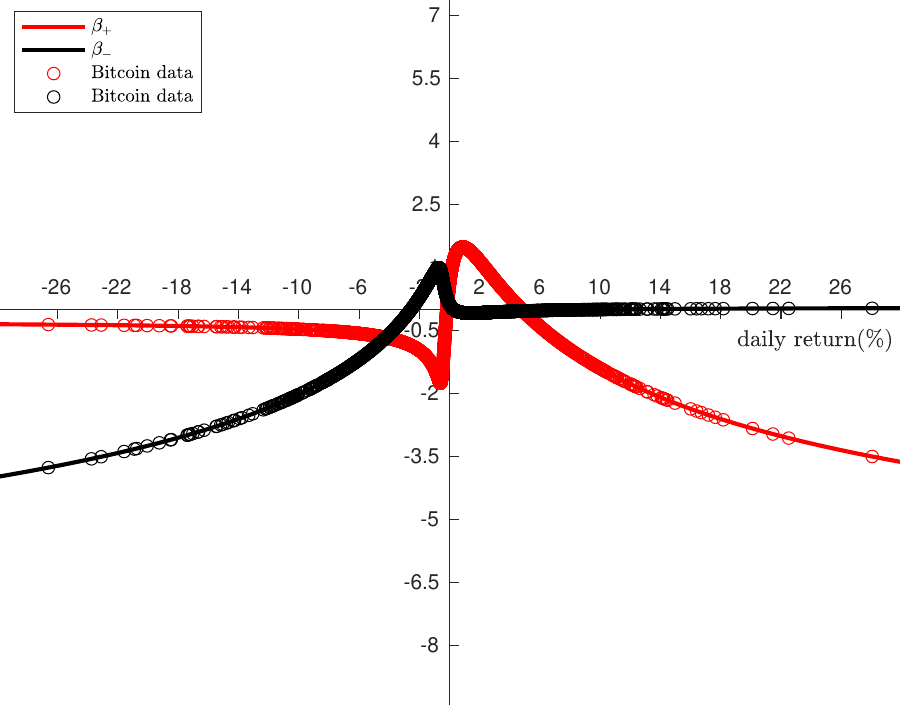}
\vspace{-0.7cm}
     \caption{Bitcoin data}
         \label{fig66}
  \end{subfigure}
\hspace{-0.45cm}
  \begin{subfigure}[b]{0.45\linewidth}
    \includegraphics[width=\linewidth]{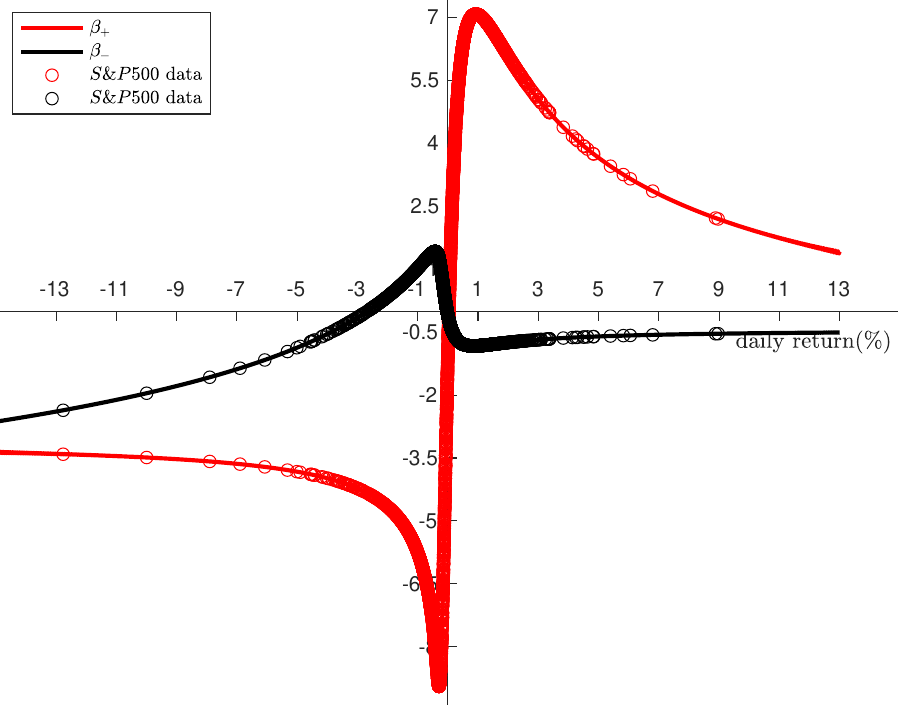}
\vspace{-0.7cm}
     \caption{$S\&P$ 500 data}
         \label{fig70}
          \end{subfigure}
\vspace{-0.7cm}
  \caption{Stability indexes ($\beta^{+}$, $\beta^{-}$): $\frac{\frac{df(x,V)}{d{\beta}}}{f(x,V)}$}
  \label{fig66a}
\vspace{-0.4cm}
\end{figure}
\begin{figure}[ht]
\vspace{-0.5cm}
    \centering
\hspace{-0.3cm}
  \begin{subfigure}[b]{0.45\linewidth}
    \includegraphics[width=\linewidth]{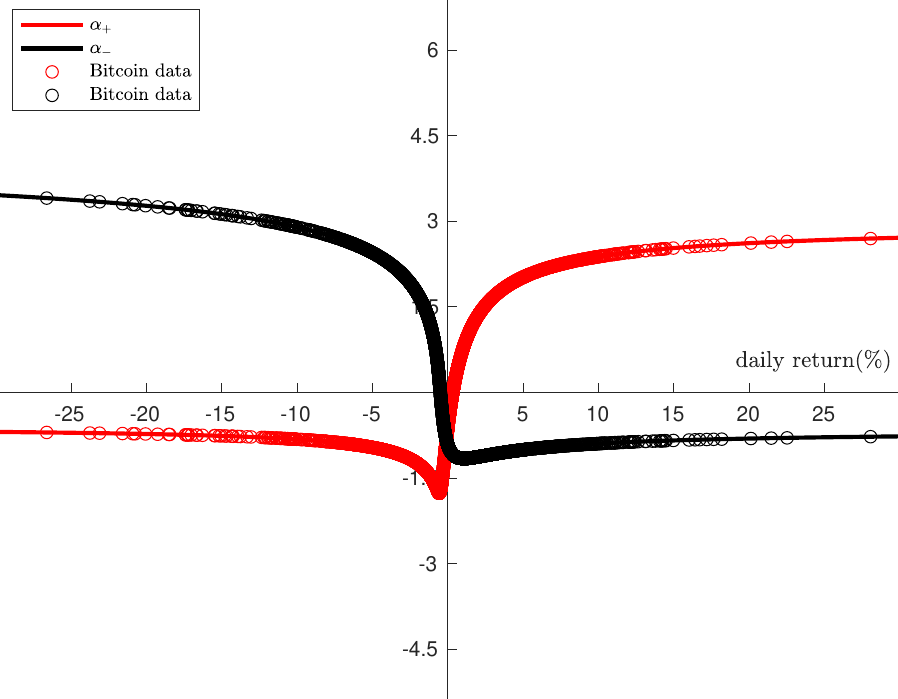}
\vspace{-0.7cm}
     \caption{Bitcoin data}
         \label{fig66}
  \end{subfigure}
\hspace{-0.4cm}
  \begin{subfigure}[b]{0.45\linewidth}
    \includegraphics[width=\linewidth]{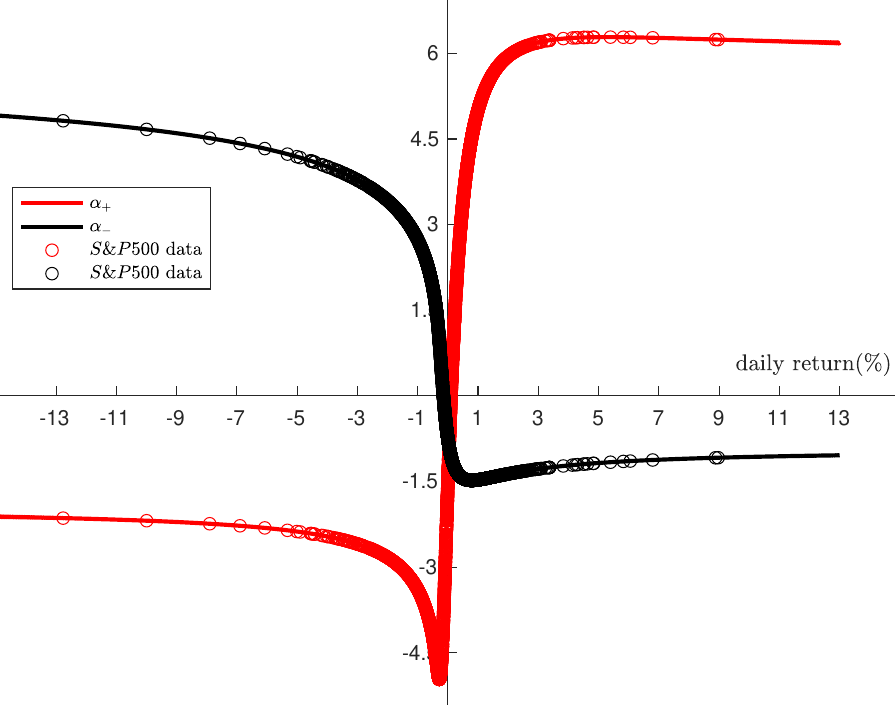}
\vspace{-0.7cm}
     \caption{$S\&P$ 500 data}
         \label{fig70}
          \end{subfigure}
\vspace{-0.7cm}
  \caption{Process intensity ($\alpha^{+}$, $\alpha^{-}$): $\frac{\frac{df(x,V)}{d{\alpha}}}{f(x,V)}$}
  \label{fig66a}
\vspace{-0.7cm}
\end{figure}
\begin{figure}[ht]
\vspace{-0.3cm}
    \centering
\hspace{-0.3cm}
  \begin{subfigure}[b]{0.45\linewidth}
    \includegraphics[width=\linewidth]{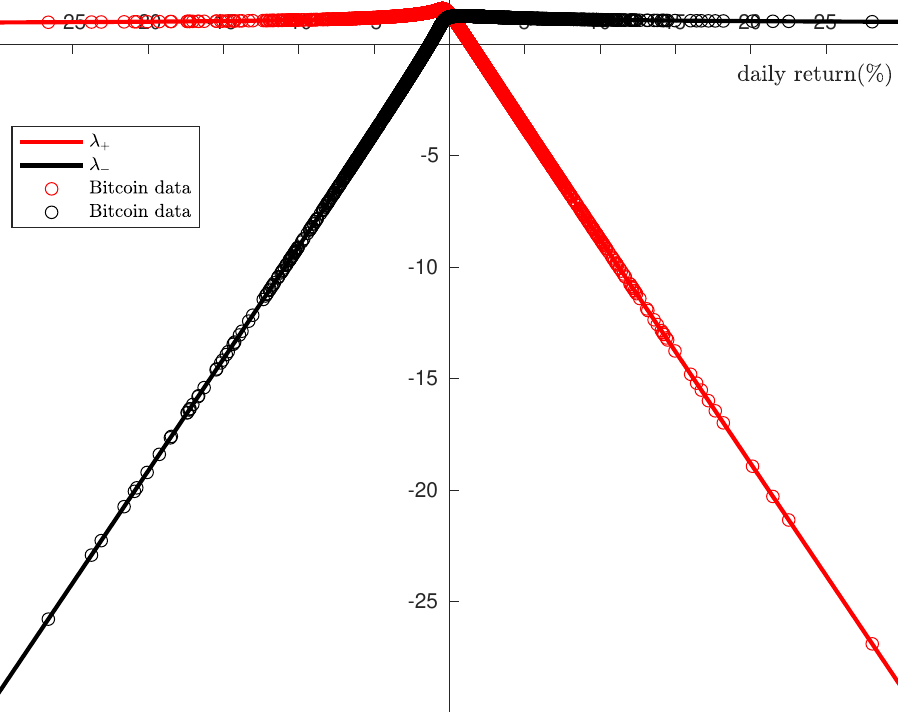}
\vspace{-0.7cm}
     \caption{Bitcoin data}
         \label{fig66}
  \end{subfigure}
\hspace{-0.45cm}
  \begin{subfigure}[b]{0.45\linewidth}
    \includegraphics[width=\linewidth]{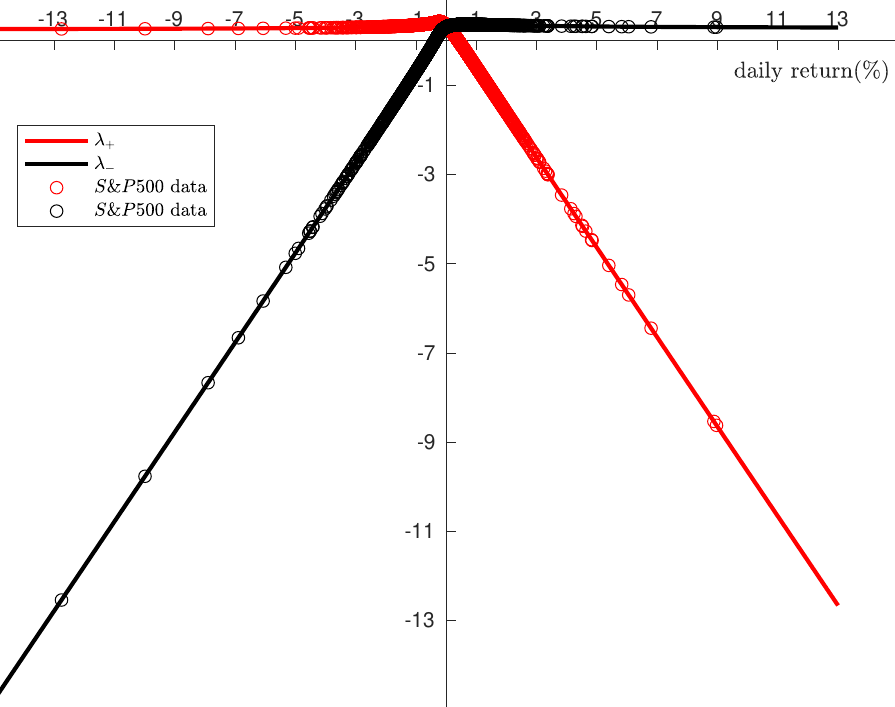}
\vspace{-0.7cm}
     \caption{$S\&P$ 500 data}
         \label{fig70}
          \end{subfigure}
\vspace{-0.7cm}
  \caption{Tail decay rate ($\lambda^{+}$, $\lambda^{-}$): $\frac{\frac{df(x,V)}{d{\lambda}}}{f(x,V)}$}
  \label{fig66a}
\vspace{-0.7cm}
\end{figure}

\section{Value-at-Risk and Average Value-at-Risk}\label{eq:an2}

\begin{table}[ht]
\vspace{-0.5cm}
 \caption{ Value-at-Risk Statistics }
\label{tab8}
\centering
\begin{tabular}{@{}c | c c | c c@{}}
\toprule
\multicolumn{1}{c|}{\textbf{$VaR_{\alpha}(X)$}} & \multicolumn{2}{c|}{\textbf{S\&P 500 index (\%)}} & \multicolumn{2}{c}{\textbf{Bitcoin (\%)}} \\ \toprule
\multirow{1}{*}{\textbf{Confidence Level ($\alpha$)}} & \multirow{1}{*}{\textbf{Empirical}} & \multirow{1}{*}{\textbf{Theoretical}} & \multirow{1}{*}{\textbf{Empirical}} & \multirow{1}{*}{\textbf{Theoretical}} \\ \toprule
\multirow{1}{*}{$0.5\%$} & \multirow{1}{*}{-4.0015} & \multirow{1}{*}{-4.6199} &\multirow{1}{*}{-16.2620} & \multirow{1}{*}{-15.0205} \\ 
\multirow{1}{*}{$1\%$} & \multirow{1}{*}{-3.2895} & \multirow{1}{*}{-3.4102} & \multirow{1}{*}{-11.9257} & \multirow{1}{*}{-12.2018} \\
\multirow{1}{*}{$2\%$} & \multirow{1}{*}{-2.5580} & \multirow{1}{*}{-2.6769} & \multirow{1}{*}{-9.7612} & \multirow{1}{*}{-9.5086} \\
\multirow{1}{*}{$3\%$} & \multirow{1}{*}{-2.1678} & \multirow{1}{*}{-2.2626} & \multirow{1}{*}{-7.7768} & \multirow{1}{*}{-7.9985} \\ 
\multirow{1}{*}{$4\%$} & \multirow{1}{*}{-1.9143} & \multirow{1}{*}{-1.9766} & \multirow{1}{*}{-6.7164} & \multirow{1}{*}{-6.9609} \\ 
\multirow{1}{*}{$5\%$} & \multirow{1}{*}{-1.7206} & \multirow{1}{*}{-1.7598} &\multirow{1}{*}{-6.1974} & \multirow{1}{*}{-6.1779} \\ 
\multirow{1}{*}{$6\%$} & \multirow{1}{*}{-1.5858} & \multirow{1}{*}{-1.5863} & \multirow{1}{*}{-5.6524} & \multirow{1}{*}{-5.5536} \\ 
\multirow{1}{*}{$7\%$} & \multirow{1}{*}{-1.4500} & \multirow{1}{*}{-1.4424} &\multirow{1}{*}{-5.1455} & \multirow{1}{*}{-5.0377} \\ 
\multirow{1}{*}{$8\%$} & \multirow{1}{*}{-1.3183} & \multirow{1}{*}{-1.3201} & \multirow{1}{*}{-4.6358} & \multirow{1}{*}{-4.6002} \\
\multirow{1}{*}{$9\%$} & \multirow{1}{*}{-1.2197} & \multirow{1}{*}{-1.2140} & \multirow{1}{*}{-4.1235} & \multirow{1}{*}{-4.2220} \\
\multirow{1}{*}{$10\%$} & \multirow{1}{*}{-1.1301} & \multirow{1}{*}{-1.1207} & \multirow{1}{*}{-3.7848} & \multirow{1}{*}{-3.8903} \\ 
\bottomrule
\end{tabular}%
\vspace{-0.6cm}
\end{table}

\bigskip

\begin{table}[ht]
\vspace{-0.5cm}
 \caption{ Average Value-at-Risk Statistics }
\label{tab9}
\centering
\begin{tabular}{@{}c | c c | c c@{}}
\toprule
  \multicolumn{1}{c|}{\textbf{$AVaR_{1-\alpha}(X)$}}& \multicolumn{2}{c|}{\textbf{S\&P 500 index (\%)}} & \multicolumn{2}{c}{\textbf{Bitcoin (\%)}} \\ \toprule
\multirow{1}{*}{\textbf{Confidence Level ($\alpha$)}} & \multirow{1}{*}{\textbf{Empirical}} & \multirow{1}{*}{\textbf{Theoretical}}  & \multirow{1}{*}{\textbf{Empirical}} & \multirow{1}{*}{\textbf{Theoretical}} \\ \toprule
\multirow{1}{*}{$0.5\%$} & \multirow{1}{*}{-5.7738} & \multirow{1}{*}{-5.3096} &\multirow{1}{*}{-19.2754} & \multirow{1}{*}{-19.2164} \\ 
\multirow{1}{*}{$1\%$} & \multirow{1}{*}{-4.6751} & \multirow{1}{*}{-4.5264} & \multirow{1}{*}{-16.6120} & \multirow{1}{*}{-16.3162} \\
\multirow{1}{*}{$2\%$} & \multirow{1}{*}{-3.7955} & \multirow{1}{*}{-3.7627} & \multirow{1}{*}{-13.6827} & \multirow{1}{*}{-13.4993} \\
\multirow{1}{*}{$3\%$} & \multirow{1}{*}{-3.3262} & \multirow{1}{*}{-3.3268} & \multirow{1}{*}{-11.9924} & \multirow{1}{*}{-11.8980} \\ 
\multirow{1}{*}{$4\%$} & \multirow{1}{*}{-3.0149} & \multirow{1}{*}{-3.0233} & \multirow{1}{*}{-10.8089} & \multirow{1}{*}{-10.7862} \\ 
\multirow{1}{*}{$5\%$} & \multirow{1}{*}{-2.7751} & \multirow{1}{*}{-2.7915} &\multirow{1}{*}{-9.9376} & \multirow{1}{*}{-9.9395} \\ 
\multirow{1}{*}{$6\%$} & \multirow{1}{*}{-2.5867} & \multirow{1}{*}{-2.6047} & \multirow{1}{*}{-9.2719} & \multirow{1}{*}{-9.2587} \\ 
\multirow{1}{*}{$7\%$} & \multirow{1}{*}{-2.4335} & \multirow{1}{*}{-2.4487} &\multirow{1}{*}{-8.7170} & \multirow{1}{*}{-8.6914} \\ 
\multirow{1}{*}{$8\%$} & \multirow{1}{*}{-2.3029} & \multirow{1}{*}{-2.3152} & \multirow{1}{*}{-8.2367} & \multirow{1}{*}{-8.2067} \\
\multirow{1}{*}{$9\%$} & \multirow{1}{*}{-2.1865} & \multirow{1}{*}{-2.1987} & \multirow{1}{*}{-7.8029} & \multirow{1}{*}{-7.7845} \\
\multirow{1}{*}{$10\%$} & \multirow{1}{*}{-2.0846} & \multirow{1}{*}{-2.0955} & \multirow{1}{*}{-7.4088} & \multirow{1}{*}{-7.4114} \\ 
\bottomrule
\end{tabular}%
\vspace{-0.6cm}
\end{table}

\end{document}